\documentclass[onecolumn]{autart}

\usepackage{graphicx,color}
\usepackage[table,usenames,dvipsnames]{xcolor}
\usepackage{amsmath,amssymb,mathrsfs,amsfonts,theorem}
\usepackage{algorithm}
\usepackage{algorithmicx}
\usepackage{algpseudocode}
\usepackage{booktabs}
\usepackage[table]{xcolor}
\usepackage{subfigure}

\newtheorem{theorem}{Theorem}[section]
\newtheorem{lemma}[theorem]{Lemma}

\newtheorem{problem}{Problem}
\newtheorem{remark}{Remark}

\newtheorem{example}{Example}

\newcommand{\transpose}{\mathsf{T}}

\renewcommand{\int}{{\mathbb{Z}}}
\newcommand{\real}{{\mathbb{R}}}

\newcommand{\map}[3]{#1: #2 \rightarrow #3}
\newcommand{\setdef}[2]{\{#1 \; : \; #2\}}
\newcommand{\subscr}[2]{{#1}_{\textup{#2}}}

\newcommand{\until}[1]{\{1,\dots,#1\}}
\newcommand{\fromto}[2]{\{#1,\dots,#2\}}
\newcommand{\Image}{\operatorname{Im}}
\newcommand{\Ker}{\operatorname{Ker}}
\newcommand{\mc}{\mathcal}

\renewcommand{\theenumi}{(\roman{enumi})}

\renewcommand{\arraystretch}{1.1} 

\def\QEDopen{{\setlength{\fboxsep}{0pt}\setlength{\fboxrule}{0.2pt}\fbox{\rule[0pt]{0pt}{1.3ex}\rule[0pt]{1.3ex}{0pt}}}}
\def\QED{\QEDopen} 

\newenvironment{pfof}[1]{\vspace{1ex}\noindent{\bf PROOF of
    #1:}\hspace{0.5em}} {\hfill\QED\vspace{1ex}\\}

\newenvironment{proof}{\vspace{1ex}\noindent{\bf PROOF:}\hspace{0.5em}}
{\hfill\QED\vspace{1ex}\\}

\DeclareSymbolFont{bbold}{U}{bbold}{m}{n}
\DeclareSymbolFontAlphabet{\mathbbold}{bbold}

\newcommand{\vect}[1]{\mathbbold{#1}}
\newcommand{\vectorones}[1][]{\vect{1}_{#1}}

\begin{document}

\begin{frontmatter}



  \title{Consensus Networks over Finite Fields}
  
  \thanks{This work was supported by NSF grants IIS-0904501 and
    CNS-1035917, and by ARO grant W911NF-11-1-0092.}

  \author[First]{F. Pasqualetti}\ead{fabiopas@engineering.ucsb.edu},
  \author[Second]{D. Borra}\ead{domenica.borra@polito.it},
  \author[First]{F. Bullo}\ead{bullo@engineering.ucsb.edu}

  \address[First]{\mbox{Center} for \mbox{Control}, \mbox{Dynamical}
    \mbox{Systems,} and \mbox{Computation}, University of California,
    Santa Barbara, USA}
  
  \address[Second]{\mbox{Dipartimento} di \mbox{Scienze}
    \mbox{Matematiche}, Politecnico di Torino, Italy}


\begin{abstract}
  This work studies consensus strategies for networks of agents with
  limited memory, computation, and communication capabilities. We
  assume that agents can process only values from a finite alphabet,
  and we adopt the framework of finite fields, where the alphabet
  consists of the integers $\{0, \dots ,p - 1\}$, for some prime
  number $p$, and operations are performed modulo $p$. Thus, we define
  a new class of consensus dynamics, which can be exploited in certain
  applications such as pose estimation in capacity and memory
  constrained sensor networks. For consensus networks over finite
  fields, we provide necessary and sufficient conditions on the
  network topology and weights to ensure convergence. We show that
  consensus networks over finite fields converge in finite time, a
  feature that can be hardly achieved over the field of real
  numbers. For the design of finite-field consensus networks, we
  propose a general design method, with high computational complexity,
  and a network composition rule to generate large consensus networks
  from smaller components. Finally, we discuss the application of
  finite-field consensus networks to distributed averaging and pose
  estimation in sensor networks.
\end{abstract} 

\end{frontmatter}

\section{Introduction}
Sensor and actuator networks have recently attracted interest from
different research communities, and, in the last years, classic
computation, control, and estimation problems have been reformulated to
conform the distributed nature of these networked systems
\cite{NAL:97,FB-JC-SM:09,MM-ME:10}. An important example is the so-called
\emph{consensus} problem, where members of a network aim to agree upon a
parameter of interest via distributed computation
\cite{FG-LS:10}. Consensus algorithms have found broad applicability in
many domains, including robotics \cite{WR-RWB-EMA:07}, estimation
\cite{LX-SB-SL:05}, and parallel computation \cite{DPB-JNT:97}.

In this work we focus on the consensus problem for networks of agents
with limited memory, computation, and communication capabilities. We
assume that agents are capable of storing, processing, and
transmitting exclusively elements from a finite and pre-specified
alphabet, and we conveniently model this situation with the formalism
of \emph{finite fields}, where the alphabet consists of a set of
integers, and operations are performed according to modular arithmetic
\cite{RL-HN:96}. We study linear consensus networks over finite fields
where, at each time instant, each agent updates its state as a
weighted combination over a finite field of its own value and those
received from its neighbors. Besides consensus in capacity and memory
constrained networks, our finite-field consensus method has broad
applicability to problems in cooperative control, networked systems,
and network coding, such as averaging, load balancing, and pose
estimation from relative measurements. Additionally, the use of a
finite alphabet for computation and communication makes our consensus
method easily implementable and resilient to communication noise.

\noindent
\textbf{Related work} Consensus algorithms have been proposed for
different network model, agents dynamics, and communication schemes.
Starting from the basic setup of time-independent network structure,
broadcast and synchronous communication, and unlimited communication
bandwidth, consensus algorithms have been proposed to cope, for
instance, with time-varying topologies \cite{YGS-LW-GX:08}, gossip and
asynchronous communication \cite{TCA-MEY-ADS-AS:09}, and communication
errors and link failures \cite{SK-JMF:09}. It has been shown that,
under mild connectivity assumptions on the interaction graphs, a
simple linear iteration ensures consensus, where at each time agents
update their state as a weighted average of the available neighboring
states \cite{LM:05}. While most of these approaches assume the
possibility of processing and transmitting real values, we consider
the more realistic case of finite communication bandwidth, possibly
due to digital communication channels, and memory
constraints. Moreover, we will show that topological conditions
ensuring consensus over real values and with real-valued
communications are not sufficient for consensus over a finite field.

A relevant body of literature deals with consensus over
\emph{quantized} communication channels, where values exchanged by the
agents are quantized according to a predefined quantization scheme,
and the proposed algorithms are resilient to quantization errors
\cite{RC-FB:06j,AN-AO-AO-JNT:09,RC-FF-PF-SZ:10,TL-MF-LX-JFZ:11,JL-RMM:12}.
In these works, although the data exchanged among robots is quantized,
agents perform computations over the field of real numbers. Thus, the
consensus problem with quantized information is related to our
problem, yet fundamentally different, since we allow agents to operate
on a finite field.

\emph{Logical} consensus has been studied in \cite{AF-EMV-AB:08} for
the purpose of intruder or event detection. In logical consensus, a
team of agents aims to coordinate their decisions via distributed
computation as a function of a set of logical (boolean) events. By
leveraging tools from cellular automata and convergence theory of
finite-state iteration maps, the focus of \cite{AF-EMV-AB:08} is on
the design of a synthesis technique for logical consensus systems. The
main differences between \cite{AF-EMV-AB:08} and this paper are as
follows. First, in logical consensus the agents state is a binary
variable, while in our work it takes value in an arbitrary finite set.
Second, in logical consensus agents are allowed to perform any logical
operation, such as $\{\text{and}, \text{or}, \text{not}\}$, while our
consensus algorithm makes use of only two mathematical operations,
namely modular addition and multiplication. Third and finally, in
logical consensus agents aim to agree upon a logical expression or
compact sets \cite{AF-ND-AB:11}, while finite-field consensus
algorithms may be used to compute a (non-boolean) function of the
initial states, such as the average.

A distributed consensus algorithm with integer communication and
computation is proposed and analyzed in \cite{AK-TB-RS:07}. With
respect to this work, and similarly to \cite{AO:12}, we make use of
\emph{modular arithmetic}, instead of standard arithmetic and,
therefore, we define a novel and complementary class of consensus
networks. The use of modular arithmetic is advantageous in several
applications, such as pose estimation from relative measurements
(Section \ref{sec:estimation}). Finally, networks based on modular
arithmetic are studied in \cite{SS-CH:12}, in the context of system
controllability and observability, in \cite{RK-MM:03}, in the context
of (linear) network coding, and in \cite{EB:59}, in the more general
context of finite dynamical systems.

\noindent
\textbf{Contributions} The contributions of this paper are
fourfold. 

First, we propose the use of modular arithmetic to design consensus
algorithms for networks of cooperative agents. Consensus networks over
finite fields are distributed, require limited, in fact finite,
memory, computation, and communication resources, and, as we show,
they exhibit finite time convergence. Thus, finite-field consensus
algorithms are suitable for capacity and memory constrained networks,
and for applications with time constraints.

Second, we thoroughly characterize convergence of consensus networks
over finite fields. We provide necessary and sufficient constructive
conditions on the network topology and weights to achieve
consensus. For instance, we show that a network achieves consensus
over a finite field if and only if the network matrix is
row-stochastic over the finite field, and its characteristic
polynomial is $s^{n-1} (s - 1)$. Additionally, we prove that the
convergence time of finite-field consensus networks is bounded by the
network cardinality, and that graph properties alone are not
sufficient to ensure finite-field consensus. Our analysis complements
the classic literature on real-valued consensus networks.

Third, we propose systematic methods to design consensus networks over
finite fields. In particular, we derive a general design method, with
high computational complexity, independent of the agents interaction
graph, and a network composition rule based on graph products to
generate large consensus networks from smaller components. We show
that networks generated by our composition rule exhibit a specific
structure, and maintain the convergence properties, including the
convergence time, of the underlying components.  Moreover, by using
our general network design method, we compute finite-field consensus
weights for some specific classes of interaction graphs, and we
provide a lower bound on the number of networks achieving consensus as
a function of the agents interaction graph and the field
characteristic.

Fourth and finally, we consider two applications in sensor networks,
namely averaging and pose estimation from relative measurements. In
the averaging problem agents aim to determine the average (over the
field of real numbers) of their initial values. We show that, under a
reasonable set of assumptions, the averaging problem can be solved
distributively and in finite time by using a finite-field average
consensus algorithm. In the pose estimation problem agents aim to
estimate their orientation based on (local) relative measurements. For
this problem we derive a distributed pose estimation algorithm based
on finite-field average consensus, and we characterize its
performance.

\noindent
\textbf{Paper organization} In Section \ref{sec:notation} we recall
some preliminary notions on fields, linear algebra, and graph theory.
Section \ref{sec: consensus} contains our setup and some preliminary
results on consensus over finite fields. In Section \ref{sec:analysis}
we provide necessary and sufficient conditions for consensus over
finite fields, and we propose several illustrative examples. In Section
\ref{sec:design} we describe methods to design consensus networks over
finite fields. Finally, Section \ref{sec:estimation} contains two
application scenarios, and Section \ref{sec:conclusion} concludes the
paper.

\section{Notation and Preliminary Concepts}\label{sec:notation}
In this section we recall some definitions and properties of algebraic
fields, linear algebra, and graph theory. We refer the interested
reader to \cite{RL-HN:96,GES:77,CDG-GFR:01} for a comprehensive
treatment of these subjects.

A \emph{field} $\mathbb{F}$ is a set of elements together with
\emph{addition} and \emph{multiplication} operations, such that the
following axioms hold:
\begin{itemize}
\item[(A1)] \emph{Closure} under addition and multiplication, that is, for
  all $a,b \in \mathbb{F}$, both $a + b \in \mathbb{F}$ and $a \cdot b
  \in \mathbb{F}$;
\item[(A2)] \emph{Associativity} of addition and multiplication, that is,
  for all $a,b,c \in \mathbb{F}$, it holds $a + (b + c) = (a + b) +
  c$ and $a \cdot (b \cdot c) = (a \cdot b) \cdot c$;
\item[(A3)] \emph{Commutativity} of addition and multiplication, that is,
  for all $a,b \in \mathbb{F}$, it holds $a + b = b +a$ and $a \cdot b
  = b \cdot a$;
\item[(A4)] \emph{Existence of additive and multiplicative identity}
  elements, that is, for all $a \in \mathbb{F}$, there exist $b,c \in
  \mathbb{F}$ such that $a + b = a$ and $a \cdot c = a$;
\item[(A5)] \emph{Existence of additive and multiplicative inverse}
  elements, that is, for all $a \in \mathbb{F}$, there exist $b,c \in
  \mathbb{F}$ such that $a + b = 0$ and $a \cdot c = 1$, with $a \neq
  0$;
\item[(A6)] \emph{Distributivity} of multiplication over addition, that is,
  for all $a,b,c \in \mathbb{F}$, it holds $a \cdot (b +c) = (a \cdot
  b) + (a \cdot c)$.
\end{itemize}
A field is finite if it contains a finite number of elements. A basic
class of finite fields are the fields $\subscr{\mathbb{F}}{p}$ with
characteristic $p$ a prime number, which consists of the set of
integers $\{0, \dots, p-1\}$, with addition and multiplication defined
as in \emph{modular arithmetic}, that is, by performing the operation
in the set of integers $\mathbb{Z}$, dividing by $p$, and taking the
remainder. Unless specified differently, all the operations listed in
this paper are performed in the field $\subscr{\mathbb{F}}{p}$.

Let $\map{a}{\subscr{\mathbb{F}}{p}^m}{\subscr{\mathbb{F}}{p}^n}$ be a
linear map between the vector spaces of dimensions $m$ and $n$,
respectively, over the field $\subscr{\mathbb{F}}{p}$. As a classical
result in linear algebra, the map $a$ can be represented by a matrix
$A$ with $n$ rows and $m$ columns, and elements from the field
$\subscr{\mathbb{F}}{p}$. The \emph{image} and \emph{kernel} of $A$
are defined as
\begin{align*}
  \Image(A) &:= \setdef{y \in \subscr{\mathbb{F}}{p}^n}{ y = A x , x \in
    \subscr{\mathbb{F}}{p}^m }, \;\;
  \Ker(A) := \setdef{x \in \subscr{\mathbb{F}}{p}^m}{ A x = 0},
\end{align*}
where additions and multiplications are performed modulo
$p$. Analogously, the \emph{pre-image} of a set of vectors $V
\subseteq \subscr{\mathbb{F}}{p}^n$ through $A$ is the set
\begin{align*}
  A^{-1} (V) := \setdef{x \in \subscr{\mathbb{F}}{p}^n}{v = A x,
    \text{ for all } v \in V}.
\end{align*}
For a matrix $A \in \subscr{\mathbb{F}}{p}^{n \times n}$, let
$\subscr{\sigma}{p} (A)$ denote the set of eigenvalues of $A$ in the
field $\subscr{\mathbb{F}}{p}$, that is,
\begin{align*}
  \subscr{\sigma}{p} (A) &=\setdef{\lambda \in
    \subscr{\mathbb{F}}{p}}{ \text{det}(\lambda I - A) = 0}.
\end{align*}
Analogously, let $\subscr{\mathbb{F}}{p}[s]$ denote the set of
polynomials with coefficients in $\subscr{\mathbb{F}}{p}$, and let
$P_{A} \in \subscr{\mathbb{F}}{p}[s]$ denote the characteristic
polynomial of $A$ over $\subscr{\mathbb{F}}{p}$.\footnote{Since the
  characteristic polynomial $\bar P_A(s) \in \real[s]$ contains only
  integer coefficients for any $A \in \subscr{\mathbb{F}}{p}^{n\times
    n}$, the characteristic polynomial $P_A(s) \in
  \subscr{\mathbb{F}}{p}[s]$ is $\sum_{i = 0}^n \text{mod}(\bar c_i,p)
  s^i$, where $\text{mod}(\cdot )$ is the modulus function, and $\bar
  c_i$ is the $i$-th coefficient of $\bar P_A(s)$.}  Notice that the
cardinality $| \subscr{\sigma}{p} (A) |$ may be strictly smaller than
the matrix dimension $n$; in other words, finite fields are not
\emph{algebraically closed}.

We conclude this section with some standard graph definitions. A
directed graph $\mathcal G = (\mathcal V, \mathcal E)$ consists of a
set of vertices $\mathcal V$ and a set of edges $\mathcal E \subseteq
\mathcal V \times \mathcal V$. An edge $(v,w) \in \mathcal E$ is
directed from vertex $w$ to vertex $v$. For a vertex $v \in \mc V$,
the set of in-neighbors of $v$ is defined as $\mathcal N^\text{in}_v =
\setdef{w \in \mc V}{ (v,w) \in \mc E }$, and the set of out-neighbors
as $\mathcal N^\text{out}_v = \setdef{w \in \mc V}{ (w,v) \in \mc
  E}$. The in-degree of $v \in \mc V$ equals $| \mc N_v^\text{in}|$,
whereas the out-degree of $v \in \mc V$ equals $| \mc
N_v^\text{out}|$. A path in $G$ is a subgraph $P = (\{v_1,\dots
,v_{k+1}\},\{e_1,\dots ,e_k\})$ such that $v_i \neq v_j$ for all $i
\neq j$, and $e_i = (v_i,v_{i+1})$ for each $i \in \until{k}$. We say
that the path starts at $v_1$ and ends at $v_{k+1}$, and, at times, we
identify a path by its vertex sequence $v_1,\dots , v_{k+1}$. A cycle
or closed path is a path in which the first and last vertex in the
sequence are the same, i.e., $v_1 = v_{k+1}$. The length of a path
(resp. cycle) equals the number of edges in the path (resp. cycle). A
directed graph is strongly (resp. weakly) connected if there exists a
directed (resp. undirected) path between any two vertices. Two
subgraphs of the same graph are disjoint if they have no common
vertices. A \emph{root} (resp. \emph{globally reachable node}) is a
vertex $v$ from which (resp. to which) there exists a directed path to
(resp. from) every vertex in the graph, including $v$ itself.
Finally, a directed graph is aperiodic if there is no integer greater
than one that divides the length of every cycle of the graph.

\section{Models of Consensus Networks over Finite Fields}\label{sec:
  consensus}
Consider a set of $n \in \mathbb{N}_{>0}$ agents and a finite field
$\subscr{\mathbb{F}}{p}$, for some prime number $p$. Let the agents
interaction be described by the directed graph $\mc G = (\mc V, \mc
E)$, where $i \in \mc V$ denotes the $i$-th agent, with $\mc V =
\until{n}$, and $(i,j) \in \mc E$ if there is a directed edge from
agent $j$ to agent $i$ (agent $i$ senses agent $j$, or, equivalently,
the behavior of agent $j$ affects agent $i$). We assume that each
agent is able to manipulate and transmit values from the finite field
$\subscr{\mathbb{F}}{p}$ according to a pre-specified protocol. We
focus on distributed protocols in which (i) each agent $i$ is
associated with a state $x_i \in \subscr{\mathbb{F}}{p}$, and (ii)
each agent updates its state as a weighted combination of the states
of its in-neighbors $\mc N_i^\text{in}$. Let $a_{ij} \in
\subscr{\mathbb{F}}{p}$ be the weight associated with the edge
$(i,j)$, and let $A = [a_{ij}]$, $A \in \subscr{ \mathbb{F}
}{p}^{n\times n}$, be the \emph{weighted adjacency matrix} of $\mc G$,
or simply \emph{network matrix}, where $a_{ij} = 0$ whenever $(i,j)
\not\in \mc E$. Let $\map{x}{\mathbb{N}_{\ge
    0}}{\subscr{\mathbb{F}}{p}^n}$ be the vector containing the agents
states. Then the evolution of the network state $x$ over time is
described by the iteration (or network)
\begin{align}\label{eq:autonomous}
  x(t+1) = A x(t),
\end{align}
where all operations are performed in the field
$\subscr{\mathbb{F}}{p}$. 

The \emph{transition graph} associated with the iteration
\eqref{eq:autonomous} over $\subscr{\mathbb{F}}{p}$ is defined as $\mc
G_A = (\mc V_A,\mc E_A)$, where, $\mc V_A = \setdef{v}{v \in
  \subscr{\mathbb{F}}{p}^n}$ and, for $v_i,v_j \in \mc V_A$, the edge
$(v_i, v_j) \in \mc E_A$ if and only if $v_j = A v_i$. It should be
observed that the transition graph contains $p^n$ vertices, and that
each vertex has unitary out-degree. Moreover, it can be shown that the
transition graph is composed of disjoint weakly-connected subgraphs,
and that each subgraph contains exactly one cycle, possibly of unitary
length \cite{RAHT:05}. Finally, each disjoint subgraph contains a
globally reachable node.
This particular structure of the transition graph
will be used to derive certain results on finite-field
consensus. Examples of transition graphs are given below in
Fig. \ref{fig:transition1} and Fig. \ref{fig:transition3}.

We say that the iteration \eqref{eq:autonomous} (or simply the network
matrix $A$) over a finite field achieves
\begin{enumerate}
\item \textbf{asymptotic consensus}, if for all initial states $x(0)
  \in \subscr{\mathbb{F}}{p}$ it holds $\lim_{t \rightarrow \infty}
  x(t) = \alpha \vectorones[]$, with $\alpha \in \subscr{\mathbb{F}}{p}$
  and $\vectorones[] = [1 \,\dots\, 1]^\transpose$;
\item \textbf{finite-time consensus}, if for all initial states $x(0)
  \in \subscr{\mathbb{F}}{p}$ there exists a finite time $T \in
  \mathbb{N}$ such that $x(T ) = x(T + \tau) = \alpha \vectorones[]$
  for all $\tau \in \mathbb{N}$, with $\alpha \in
  \subscr{\mathbb{F}}{p}$ and $\vectorones[] = [1 \,\dots\,
  1]^\transpose$.
\end{enumerate}

Consensus networks with real-valued weights and states have been
extensively studied \cite{LM:05,WR-RWB:05,FG-LS:10}. In this
work we show that real-valued consensus networks and finite-field
consensus networks exhibit different features, and particular care
needs to be taken to ensure the desired properties over finite
fields. It is clear from the above definitions that finite-time
consensus implies asymptotic consensus. We next show that the converse
is also true.

\begin{theorem}{\bf \emph{(Asymptotic consensus implies
      finite-time consensus)}}\label{thm:convergence}
  The iteration \eqref{eq:autonomous} over the field
  $\subscr{\mathbb{F}}{p}$ achieves asymptotic consensus only if it
  achieves finite-time consensus.
\end{theorem}
\begin{proof}
  Let $\mc G_A = (\mc V_A, \mc E_A)$ be the transition graph
  associated with the iteration \eqref{eq:autonomous}. Notice that the
  state trajectory $x$ of \eqref{eq:autonomous} coincides with a path
  on $\mc G_A$ starting from the vertex $v_0 = x(0)$. Let $\mc C
  \subset \mc V_A$ be the set of consensus vertices, that is, $\mc C =
  \setdef{v}{v \in \mc V_A, v = \alpha \vectorones[], \alpha \in
    \subscr{\mathbb{F}}{p}}$. Suppose that the iteration
  \eqref{eq:autonomous} achieves consensus on the value $v_c \in \mc
  C$. Since the vertex set $\mc V_A$ is finite, the distance between
  $v_0$ and $v_c$ is also finite. Consequently, a consensus vertex is
  reached with a path on $\mc G$ of finite length, that is, a finite
  number of iterations in \eqref{eq:autonomous} are sufficient to
  achieve consensus.
\end{proof}
Following Theorem \ref{thm:convergence}, iterations over finite fields
either achieve consensus in finite time, or they are not
convergent. In view of this result, in what follows we simply use
consensus instead of finite-time and asymptotic consensus. Differently
from finite-field consensus networks, consensus networks over the
field of real numbers usually converge asymptotically. An exception is
constituted by the class of de Bruijn graphs, which have been shown to
yield finite-time consensus over the field of real numbers
\cite{JCD-RC-SZ:07}. On the other hand, de Bruijn graphs rely on a
specific interaction graph, while finite-field consensus networks
include a much broader class of interaction graphs. We conclude this
section with a simple result. A matrix $A$ over the field
$\subscr{\mathbb{F}}{p}$ is \emph{nilpotent} if $A^n=0$ and is
\emph{row-stochastic} if $A\vectorones[]=\vectorones[]$.

\begin{lemma}{\bf \emph{(Finite-field consensus
      matrices)}}\label{lemma:stochastic} Consider the iteration
  \eqref{eq:autonomous} over the field $\subscr{\mathbb{F}}{p}$. If
  consensus is achieved, then $A$ is either nilpotent or row-stochastic.
  \label{lemma: stochastic}
\end{lemma}
\begin{proof}
  Since $A$ achieves consensus, it follows from Theorem
  \ref{thm:convergence} that $A^t x(0) = A^{t+1} x(0) = \alpha
  \vectorones[]$ for some $\alpha \in \subscr{\mathbb{F}}{p}$, for all
  $x(0)$, and for all $t \ge T$, $T \in \mathbb{N}$. Then $A \alpha
  \vectorones[] = \alpha \vectorones[]$, from which we conclude that
  either $A \vectorones[] = \vectorones[]$ ($A$ is row-stochastic) or
  $\alpha = 0$ for all initial states $x(0)$ ($A$ is nilpotent).
\end{proof}
As for the case of real-valued consensus, we limit our attention to
row-stochastic network matrices. Although consensus is trivially
achieved whenever the network matrix is nilpotent, this case is of
limited interest because the consensus value is the origin
independently of the agents initial states.

\section{Analysis of Consensus Networks over Finite
  Fields}\label{sec:analysis} Conditions for consensus in real-valued
networks have been deeply investigated in the last years
\cite{LM:05,WR-RWB:05,FG-LS:10}. For instance, sufficient conditions
ensuring real-valued consensus are that the network matrix $A$ is
row-stochastic and that the associated directed graph is strongly connected
and aperiodic.
The following example shows that graph-theoretic properties are not
sufficient for an iteration over a finite field to achieve consensus.

\begin{example}{\bf \emph{(Graph properties are not sufficient for
      finite-field consensus)}}\label{example:consensus}
  Consider a fully connected network with three agents over the field
  $\mathbb{F}_3$. Consider the network matrices
  \begin{align*}
    A_1 &= 
    \begin{bmatrix}
      2 & 1 & 1\\
      2 & 1 & 1\\
      2 & 1 & 1
    \end{bmatrix},
    A_2 = 
    \begin{bmatrix}
      2 & 1 & 1\\
      1 & 2 & 1\\
      1 & 2 & 1
    \end{bmatrix},
    \text{ and }
    A_3 = 
    \begin{bmatrix}
      2 & 1 & 1\\
      1 & 2 & 1\\
      1 & 1 & 2
    \end{bmatrix}.
  \end{align*}
  Notice that $A_1$, $A_2$, and $A_3$ are row-stochastic and their
  interaction graph is fully connected. It can be verified that over
  the field $\subscr{\mathbb{F}}{3}$ only the network matrix $A_1$
  achieves consensus, while $A_2$ and $A_3$ exhibit oscillatory
  dynamics for certain initial conditions. An example of oscillatory
  dynamics generated by $A_3$ is reported in Table \ref{table:a2}.

  \begin{table}[!htbp]
    \caption{Sample state trajectory for the matrix $A_3$ in Example
      \ref{example:consensus}.}
    \renewcommand{\arraystretch}{1.2}
    \begin{center}
      \resizebox{.9\linewidth}{!}{
        \begin{tabular}
          {
            >{\centering \arraybackslash \columncolor{black!05!white}}
            p{.16\columnwidth}
            >{\centering \arraybackslash} p{0.16\columnwidth} 
            >{\centering \arraybackslash \columncolor{black!05!white}}
            p{0.16\columnwidth}
            >{\centering \arraybackslash} p{0.15\columnwidth} 
            >{\centering \arraybackslash \columncolor{black!05!white}}
            p{0.16\columnwidth}
            >{\centering \arraybackslash} p{0.16\columnwidth} 
            >{\centering \arraybackslash \columncolor{black!05!white}}
            p{0.16\columnwidth}
          }
          \toprule
          $x(0)$ & $x(1)$ & $x(2)$ & $x(3)$ & $x(4)$ & $x(5)$ & $x(6)$\\
          \midrule
          \midrule
          1 & 2 & 0 & 1 & 2 & 0 & 1\\
          0 & 1 & 2 & 0 & 1 & 2 & 0\\
          0 & 1 & 2 & 0 & 1 & 2 & 0\\
          \bottomrule
        \end{tabular}}
      \label{table:a2}
    \end{center}
  \end{table}
{\hfill\QED}
\end{example}

As shown in Example \ref{example:consensus}, graph properties of the
network matrix are not sufficient to guarantee consensus for
iterations over a finite field (in short, \emph{finite-field
  consensus}). Indeed, although the considered network matrices
feature the same connectivity properties, only one of them achieves
finite-field consensus. In what follows we provide finite-field
consensus conditions based on the algebraic properties of the network
matrix, and on the topological properties of its transition graph.

The dynamic behavior of an iteration over a finite field is entirely
described by its associated transition graph. The next theorem
provides a necessary and sufficient condition for finite-field
consensus based on the transition graph.

\begin{theorem}{\bf \emph{(Transition graph of a consensus
      network)}}\label{thm:transition}
  Consider the iteration \eqref{eq:autonomous} over the field
  $\subscr{\mathbb{F}}{p}$ with row-stochastic matrix $A$, and let
  $\mc G_A = (\mc V_A, \mc E_A)$ be its associated transition
  graph. The following statements are equivalent:
  \begin{enumerate}
  \item the iteration \eqref{eq:autonomous} achieves consensus, and
  \item the transition graph $\mc G_A$ contains exactly $p$ cycles,
    corresponding to the unitary cycles around the vertices $\alpha
    \vectorones[]$, $\alpha\in\fromto{0}{p-1}$.
  \end{enumerate}
\end{theorem}

  \begin{figure}
    \centering
    \includegraphics[width=1\columnwidth]{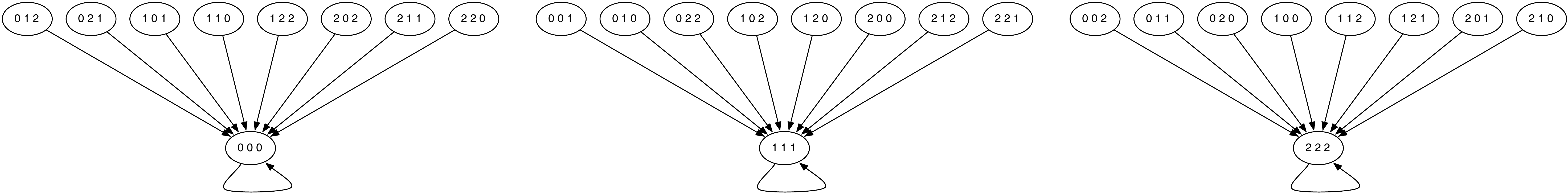}
    \caption{Transition graph $\mc G_{A_1}$ associated with the matrix
      $A_1 \in \mathbb{F}_3^{3\times 3}$ in Example
      \ref{example:consensus}. Since $\mc G_{A_1}$ contains exactly
      $3$ cycles corresponding to the self-loops around the consensus
      vertices, the network matrix $A_1$ achieves consensus (see
      Theorem \ref{thm:transition}).}
    \label{fig:transition1}
\end{figure}

  \begin{figure}
    \centering
    \includegraphics[width=.8\columnwidth]{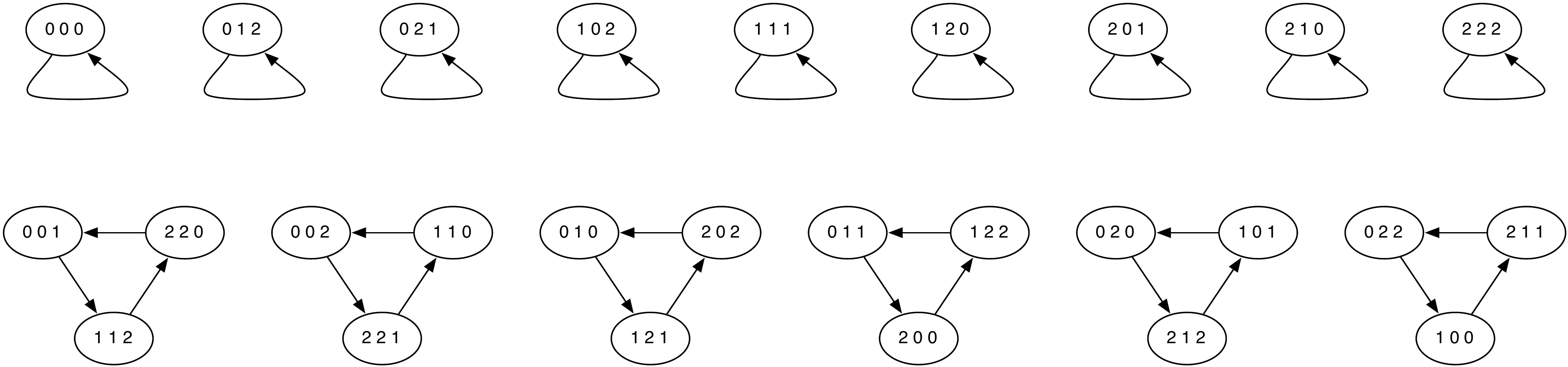}
    \caption{Transition graph $\mc G_{A_3}$ associated with the matrix
      $A_3 \in \mathbb{F}_3^{3\times 3}$ in Example
      \ref{example:consensus}. Since $\mc G_{A_3}$ contains more than
      $3$ cycles, the network matrix $A_3$ does not achieve consensus
      (see Theorem \ref{thm:transition}). The oscillatory state
      trajectory in Table \ref{table:a2} corresponds to the bottom
      right cycle in this figure.}
    \label{fig:transition3}
\end{figure}

\begin{example}{\bf \emph{(Transition graph of a consensus
      network)}}\label{thm:char_pol}
  The transition graphs associated with the matrices $A_1$ and $A_3$
  in Example \ref{example:consensus} over the field $\mathbb{F}_3$ are
  reported in Fig. \ref{fig:transition1} and
  Fig. \ref{fig:transition3}, respectively. As previously discussed,
  and as predicted by Theorem \ref{thm:transition}, the matrix $A_1$
  achieves consensus, while the matrix $A_3$ does not.   
  {\hfill\QED}
\end{example}

\begin{pfof}{Theorem \ref{thm:transition}}

  \noindent
  \textit{(i) $\implies$ (ii)} Since the iteration achieves consensus,
  it follows from Lemma \ref{lemma: stochastic} that $A \vectorones[]
  = \vectorones[]$. Hence, the transition graph contains $p$ unitary
  cycles corresponding to the vertices $\alpha \vectorones[] \in \mc
  V_A$, with $\alpha \in\subscr{\mathbb{F}}{p}$. Suppose by
  contradiction that there exists an additional cycle $C$, and notice
  that the vertices $\alpha \vectorones[]$, with $\alpha
  \in\subscr{\mathbb{F}}{p}$, cannot be contained in $C$ since the
  out-degree of each vertex in the transition graph is exactly one
  (the transition graph is determined by the linear map $A$). Thus,
  there exists a trajectory along $C$ that does not converge to
  consensus, which contradicts the initial hypothesis.

  \noindent
  \textit{(ii) $\implies$ (i)} Notice that a state trajectory of the
  iteration \eqref{eq:autonomous} is in bijective correspondence with
  a path on the transition graph $\mc G_A$. Suppose that transition
  graph $\mc G_A$ contains exactly $p$ unitary cycles located at the
  vertices $\alpha \vectorones[] \in \mc V_A$, with $\alpha
  \in\subscr{\mathbb{F}}{p}$. Then, since each vertex in the
  transition graph has unitary out-degree, every (sufficiently long)
  path in $\mc G_A$ eventually reaches one of the cycles, and,
  consequently, every state trajectory converges to a consensus state.
\end{pfof}

Theorem \ref{thm:transition} provides a necessary and sufficient
condition for finite-field consensus based on the transition
graph. From condition (ii) in Theorem \ref{thm:transition} and the
fact that each vertex in the transition graph has unitary out-degree,
we also note that the transition graph of a consensus matrix is
composed of $p$ disjoint weakly-connected subgraphs. Moreover, by
means of \cite[Proposition 3.4]{RAHT:05}, it can be shown that
disjoint subgraphs have the same graph topology.

A verification of the convergence condition in Theorem
\ref{thm:transition} may be prohibitive for large networks, because
the size of the transition graph grows exponentially with the number
of agents in the network (the transition graph contains $p^n$ vertices
and $p^n$ edges, since each vertex has unitary out-degree). In what
follows we shall derive consensus conditions based on the network
matrix instead of its transition graph. Consider the \emph{inverse
  recursion}
\begin{align}\label{eq:recursion}
  \mc S^{t+1}_\alpha = A^{-1} (\mc S^t_\alpha),
\end{align}
where $\mc S^t_\alpha\subset\subscr{\mathbb{F}}{p}^n$ for all times
$t$ and $\mc S^0_\alpha = \left\{\alpha \vectorones[]\right\}$. Notice
that the inverse recursion defines a sequence of sets, and that the
set $\mc S^t_\alpha$ contains the initial states converging to the
consensus value $\alpha \vectorones[]$ in at most $t$ iterations. We
say that the recursion \eqref{eq:recursion} is convergent with
limiting set $\mc S_\alpha$ if there exists $T < n$ satisfying $\mc
S_\alpha = \mc S_\alpha^{T} = \mc S_\alpha^{T + 1}$. The following
theorem exploits the inverse recursion and the structure of the
transition graph to characterize finite-field consensus.

\begin{theorem}{\bf \emph{(Recursion subspaces of a consensus
      network)}}\label{thm:recursion} Consider the iteration
  \eqref{eq:autonomous} over the field $\subscr{\mathbb{F}}{p}$ with
  row-stochastic matrix $A$. The following statements are equivalent:
  \begin{enumerate}
  \item the iteration \eqref{eq:autonomous} achieves consensus,
  \item there exists $\alpha \in \subscr{\mathbb{F}}{p}$ such that the
    recursion \eqref{eq:recursion} is convergent and the limiting set
    $S_\alpha$ satisfies $| \mc S_{\alpha} | = p^{n-1}$, and
  \item for all $\alpha \in \subscr{\mathbb{F}}{p}$ the recursion
    \eqref{eq:recursion} is convergent and each limiting set $S_\alpha$
    satisfies $| \mc S_{\alpha} | = p^{n-1}$.
  \end{enumerate}
\end{theorem}

\begin{example}{\bf \emph{(Inverse recursion for a finite-field
      consensus network)}}\label{example:recursion}
  For the matrix $A_1 \in \mathbb{F}_3^{3 \times 3}$ in Example
  \ref{example:consensus}, the set $\mc S_1$ generated by the inverse
  recursion \eqref{eq:recursion} is
  \begin{align*}
    \mc S_1 =
    \left\{
      \begin{bmatrix}
        1 \\ 1 \\ 1
      \end{bmatrix}\!,\!
      \begin{bmatrix}
        0 \\ 0 \\ 1
      \end{bmatrix}\!,\!
      \begin{bmatrix}
        0 \\ 1 \\ 0
      \end{bmatrix}\!,\!
      \begin{bmatrix}
        0 \\ 2 \\ 2
      \end{bmatrix}\!,\!
      \begin{bmatrix}
        1 \\ 0 \\ 2
      \end{bmatrix}\!,\!
      \begin{bmatrix}
        1 \\ 2 \\ 0
      \end{bmatrix}\!,\!
      \begin{bmatrix}
        2 \\ 0 \\ 0
      \end{bmatrix}\!,\!
      \begin{bmatrix}
        2 \\ 1 \\ 2
      \end{bmatrix}\!,\!
      \begin{bmatrix}
        2 \\ 2 \\ 1
      \end{bmatrix}
    \right\}.
  \end{align*}
  Because $|\mc S_1| = 3^2$, the network matrix $A_1$ achieves
  consensus due to Theorem \ref{thm:recursion}. Instead, for the
  network matrix $A_3 \in \mathbb{F}_3^{3 \times 3}$ in Example
  \ref{example:consensus}, the inverse recursion yields $\mc S_1 = \{
  \vectorones[] \}$, so that $A_3$ does not achieve consensus.
  {\hfill\QED}
\end{example}

\begin{pfof}{Theorem \ref{thm:recursion}} Consider the transition
  graph $\mc G_A = (\mc V , \mc E)$, and define the reverse graph
  $\bar{\mc G}_A = (\mc V, \bar{\mc E})$, where $(i,j) \in \bar{\mc
    E}$ if and only if $(j,i) \in \mc E$.  Notice that the recursion
  \eqref{eq:recursion} is convergent if and only if
  $\bar{\mc G}_A$ contains no cycle of length greater than $1$
  reachable from $\alpha \vectorones[]$. Recall from \cite[Theorem
  1]{RAHT:05} that $\mc G_A$ (resp. $\bar{\mc G}_A$) is obtained as
  the graph product of a tree by a set of cycles. Hence, the graph
  $\mc G_A$ (resp. $\bar{\mc G}_A$) is composed of disjoint
  weakly-connected subgraphs, and disjoint subgraphs have the same
  structure. From this argument we conclude that (ii) and (iii) are
  equivalent.

  \noindent
  \textit{(i) $\implies$ (ii)} Since $A$ achieves consensus, the graph
  $\mc G_A$ contains exactly $p$ unitary cycles corresponding to the
  consensus vertices (see Theorem \ref{thm:transition} and
  Fig. \ref{fig:transition1}). By \cite[Theorem 1]{RAHT:05}, the above
  reasoning, and the fact that $A$ achieves consensus, it follows that
  $|\mc S_0| + \dots + |\mc S_{p-1}| = p^n$, and that $|\mc S_\alpha|
  = p^{n - 1}$ for all $\alpha \in \subscr{\mathbb{F}}{p}$.

  \noindent
  \textit{(ii) $\implies$ (i)} Since $A \vectorones[] =
  \vectorones[]$, the transition graph contains $p$ cycles of unitary
  length located at the consensus vertices. Let the recursion
  \eqref{eq:recursion} be convergent for some $\alpha \in
  \subscr{\mathbb{F}}{p}$. From \cite[Theorem 1]{RAHT:05}, the graph
  $\mc G_A$ contains $p$ identical, disjoint, weakly-connected
  subgraphs, each one terminating in a consensus vertex. Since $| \mc
  S_\alpha| = p^{n - 1}$, it follows that consensus is achieved from
  $p^n$ states (every initial state), which concludes the proof.
\end{pfof}

According to Theorem \ref{thm:recursion}, the convergence of the
network \eqref{eq:autonomous} can be determined by iterating the
inverse recursion \eqref{eq:recursion} for some $\alpha \in
\subscr{\mathbb{F}}{p}$. This computation does not require analyzing
the transition graph. Our last and most explicit condition for
finite-field consensus is based upon the characteristic polynomial of
the network matrix (computed over the finite field).

\begin{theorem}{\bf \emph{(Characteristic polynomial of a consensus
      network)}}\label{thm:char_pol} Consider the iteration
  \eqref{eq:autonomous} over the finite field $\subscr{\mathbb{F}}{p}$ with
  row-stochastic matrix $A$. The following statements are equivalent:
  \begin{enumerate}
  \item the iteration \eqref{eq:autonomous} achieves consensus, and
  \item $P_A (s) = s^{n - 1}(s - 1)$.
  \end{enumerate}
\end{theorem}

\begin{example}{\bf \emph{(Characteristic polynomial of consensus
      matrices)}}\label{example:char_pol}
  Consider the network matrices in Example \ref{example:consensus}
  over the field $\mathbb{F}_3$. It can be verified that
  \begin{align*}
    P_{A_1}(s) &= s^2 (s - 1), \;\;
    P_{A_2}(s) = s (s^2 - 2s + 1),\text{ and }\;
    P_{A_3}(s) = s^3 - 1.
  \end{align*}
  As predicted by our previous analysis and by Theorem
  \ref{thm:char_pol}, only the network matrix $A_1$ achieves
  consensus.
  {\hfill\QED}
\end{example}

The proof of this theorem is postponsed to the Appendix.  Theorem
\ref{thm:char_pol} is equivalently restated as follows: $A$ achieves
finite-field consensus if and only if $\subscr{\sigma}{p} (A) =
\{1,0,\dots,0\}$. In other words, the eigenvalues of a finite-field
matrix achieving consensus are all contained in the considered finite
field and, consequently, every finite-field matrix achieving consensus
can be represented in Jordan canonical form via a similarity
transformation; see~\cite{CDM:01} and~\cite[Theorem 3.5]{PS:10}. We
conclude this section by characterizing the convergence value of a
finite-field consensus network.

\begin{theorem}{\bf \emph{(Finite-field consensus time and
      value)}}\label{thm:convergence_value}
  Consider the iteration \eqref{eq:autonomous} over the finite field
  $\subscr{\mathbb{F}}{p}$ with row-stochastic matrix $A$ and with initial
  state $x(0)\in\subscr{\mathbb{F}}{p}^n$. Assume the iteration
  \eqref{eq:autonomous} achieves consensus.
  Let $T < n$ denote the dimension of the largest Jordan block
  associated with the eigenvalue $0$. Let
  $\pi\in\subscr{\mathbb{F}}{p}^n$ be the unique eigenvector
  satisfying $\pi A = \pi$ and $\pi \vectorones[] = 1$. Then
  \begin{align*}
    A^T = \vectorones[] \pi,
  \end{align*}
  so that consensus is achieved at the value $\pi x(0)$ after $T$
  iterations. Moreover, the $i$-th component of $\pi$ is nonzero only
  if the $i$-th vertex of the directed graph associated with $A$ is a
  root.
\end{theorem}

\begin{proof}
  Since $A$ achieves consensus, we have $\subscr{\sigma}{p} (A) =
  \{1,0,\dots,0\}$, and $A$ admits a Jordan canonical form $J_A =
  V^{-1} A V$ over $\subscr{\mathbb{F}}{p}$. Moreover, the matrix $A$
  converges in $T < n$ iterations. The next part of the proof follows
  the reasoning in \cite[Theorem 3]{ROS-RMM:03c}. Let the first column
  of $T$ be $\vectorones[]$, and let the matrix $J_A$ have only zero
  elements, except for the unitary entry in position $(1,1)$. Since
  $V^{-1} A = J_A V^{-1}$, the first row of $V^{-1}$, say $\pi$,
  satisfies $\pi A = \pi$. Then $A^T = V J_A^T V^{-1} = \vectorones[]
  \pi$. Since $A \vectorones[] = \vectorones[]$, it follows that
  $\vectorones[] = A^T \vectorones[] = \vectorones[] \pi
  \vectorones[]$, and consequently $\pi \vectorones[] = 1$. To show
  the last statement, let $\mc G$ be the directed graph associated
  with $A$, and let $i$ be a vertex of $\mc G$. Assume that $i$ is not
  a root of $\mc G$, and let the initial state $x(0)$ be all zeros,
  except for the $i$-th component. Since $i$ is not a root, there
  exists a node $j$ which is not reachable from $i$, and,
  consequently, the value of the $j$-th agent is not affected by the
  $i$-th agent. Since $A$ achieves consensus for all initial states,
  the $j$-th entry of $\vectorones \pi x(0)$ needs to be zero, from
  which the statement follows.
\end{proof}

Observe that Theorem \ref{thm:convergence_value} is not a direct
consequence of the theory of non-negative matrices over the field of
real numbers \cite{CDM:01}. In fact, if regarded as a real-valued
matrix, a finite-field consensus matrix is generally unstable.

\section{Design of Consensus Networks over Finite
  Fields}\label{sec:design}
In Section \ref{sec:analysis} we characterize the convergence of
consensus networks over finite fields. With respect to consensus
networks over the field of real numbers, finite-field consensus
networks require less computational effort and communication
bandwidth, and they converge in a finite number of iterations. On the
other hand, convergence conditions for finite-field consensus networks
depend on the numerical entries of the network matrix (Theorem
\ref{thm:char_pol}), and not only on the connectivity properties of
the underlying graph as in the case of consensus networks with real
values. For this reason, the design of finite-field consensus networks
deserves particular attention. In this section we describe methods to
design finite-field consensus networks, and we discuss their
limitations. First, we elaborate on Theorem \ref{thm:char_pol} to
describe a general design method, and we propose a solution for
particular interaction graphs. Then we present a scalable composition
rule to generate finite-field consensus networks from smaller
consensus components.

\subsection{Network design via characteristic polynomial}
The objective of this section is to design a finite-field consensus
matrix $A$ whose sparsity patter is compatible with a given agents
interaction graph $\mc G = (\mc V, \mc E)$, that is, to design the
network matrix $A = [a_{ij}]$, $a_{ij} \in \subscr{\mathbb{F}}{p}$,
where $p$ is a given prime number, and $a_{ij} \neq 0$ only if $(i,j)$
is an edge of $\mc G$. Recall from Theorem \ref{thm:char_pol} that the
network matrix $A$ achieves consensus if and only if its
characteristic polynomial is $P_A(s) = s^n (s - 1)$ and $A \vectorones[]
= \vectorones[]$. Since two polynomials are equal if and only if all
coefficients are equal, the entries of $A$ can be determined by
simultaneously solving the following equations:
\begin{align}\label{eq:diophantine}
  \left\{
  \begin{array}{ll}
    A &\in \subscr{\mathbb{F}}{p}^{n \times n}, \text{ with } a_{ij}
    = 0 \text{ if } (i,j) \not\in \mc E,\\
    \vectorones[] &= A \vectorones[],\\
    1 &= -c(A,n-1),\\
    0 &= c(A,j), \;\; j \in\fromto{0}{n-2},
  \end{array}
  \right.
\end{align}
where $c(A,d)$ is the coefficient of the monomial of degree $d$ in the
(parametric) characteristic polynomial $P_A$. 

\begin{example}{\bf \emph{(Design for networks with $2$
      agents)}}\label{remark:complete_design}
  Let $A =
  \begin{bmatrix}
    a_{11} & a_{12} \\
    a_{21} & a_{22}
  \end{bmatrix} 
  \in \subscr{\mathbb{F}}{p}^{2 \times 2}$, and notice that $P_A(s) =
  s^2 - (a_{11} + a_{22} ) s + a_{11}a_{22}- a_{12}a_{21}$. It follows
  from \eqref{eq:diophantine} that $A$ achieves consensus if and only
  if $a_{11} = 1 - a_{12}$, $a_{22} = 1 - a_{21}$, $a_{12} + a_{21} =
  1$, and, consequently, if and only if $A = \begin{bmatrix}
    1-\alpha &\quad \alpha\\
    1-\alpha &\quad \alpha
  \end{bmatrix}$ for some $\alpha \in \subscr{\mathbb{F}}{p}$.
  {\hfill\QED}
\end{example}

Notice that the system of equations \eqref{eq:diophantine} contains
non-linear, multivariate, polynomial equations, where the unknown
variables are the entries of the network matrix $A$, and that a
solution is required over the finite field
$\subscr{\mathbb{F}}{p}$. The problem of solving systems of
multivariate polynomial equations over finite fields is \emph{NP-hard}
\cite{ASF-YY:79,MRG-DSJ:79}, and it is one of the important research
problems in cryptography and information security
\cite{JD-JEG-DS:06}. Besides enumerating all possible solution
candidates,\footnote{If $m$ is the number of free entries in $A$, a
  brute-force solution to \eqref{eq:diophantine} requires computing
  all $p^m$ possible matrices $A$ compatible with the interaction
  graph $\mc G$.} several solution techniques have been proposed over
the last years, see for instance
\cite{LB-JCF-LP:09,NC-AK-JP-AS:00}. We next provide a condition on the
agents interaction graph for the existence of a finite-field network
achieving consensus.

\begin{theorem}{\bf \emph{(Existence of finite-field consensus
      matrices)}}\label{thm:existence}
  Let $\mc G = (\mc V, \mc E)$ be the directed agents interaction
  graph, with $| \mc V| = n$ and $|\mc E| = m$. Assume that $\mc G$
  contains a root, and that $m > \frac{n^2 + n}{2}$. Then the system
  of equations \eqref{eq:diophantine} admits $N \ge p^{\frac{2m - n^2
      -n}{2}}$ solutions over the field $\subscr{\mathbb{F}}{p}$,
  where $N$ is divisible by the characteristic $p$. In other words,
  there exist $N$ network matrices $A = [a_{ij}]$ achieving consensus,
  with $a_{ij} \in \subscr{\mathbb{F}}{p}^{n \times n}$ and $a_{ij} =
  0$ if $(i,j) \not\in \mc E$.
\end{theorem}
\begin{proof}
  Let $c(A,d)$ be the coefficient of the monomial of degree $d$ in the
  parametric characteristic polynomial $P_A$, with $A = [a_{ij}]$ and
  $a_{ij} = 0$ if $(i,j) \not\in \mc E$. Define the polynomial $f_i
  \in \subscr{\mathbb{F}}{p} [a_{ij}]$ with $(i,j) \in \mc E$ as: $f_i = f_i
  (a_{ij})= \sum_{j = 1}^n a_{i j} -1$ for $i \in \until{n}$, $f_{n+1}
  = 1 + c(A,n-1)$, and $f_i = c(A,2n - i)$ for $i \in
  \fromto{n+2}{2n}$. Notice that a solution to \eqref{eq:diophantine},
  and hence a finite-field consensus matrix compatible with $\mc G$,
  can be computed by simultaneously solving the equations $f_i = 0$
  for $i \in \until{2n}$. Observe that $\text{deg}(f_i) = 1$ for $i
  \in \until{n}$, and $\text{deg} (f_{i}) \le i - n$ for $i \in
  \fromto{n+1}{2n}$. Then $\sum_{i = 1}^{2n} \text{deg} (f_i) \le n +
  \frac{n(n-1)}{2} \le \frac{n^2 + n}{2}$. Since $m > \frac{n^2 +
    n}{2}$ by assumption, we conclude from \cite[Theorem
  6.8]{RL-HN:96} that the number of simultaneous solutions $N$ to the
  equations $f_i = 0$, $i \in \until{2n}$, in the field
  $\subscr{\mathbb{F}}{p}$ is divisible by the characteristic $p$. To
  show that $N \ge p^{\frac{2 m - n^2 -n}{2}}$, we next construct a
  network matrix $A$ achieving consensus and compatible with the
  interaction graph $\mc G$, and then employ \cite[Theorem
  6.11]{RL-HN:96}. Let $v$ be a root of $\mc G$, and let $S = (\mc V,
  \mc E_S)$ be a rooted spanning tree of $\mc G$ with root $v$, and
  with $(v,v) \in \mc E_S$ and $(i,i) \not\in \mc E_S$ for $i \neq
  v$. Relabel the vertices $\mc V$ according to their distance from
  the root $v$. Define the matrix $A = [a_{ij}]$, where $a_{ij} =1$ if
  $(i,j) \in \mc E_S$, and $a_{ij} = 0$ otherwise. Notice that $A \in
  \subscr{\mathbb{F}}{p}^{n \times n}$ is triangular and
  row-stochastic, and that its diagonal elements are
  $\{1,0,\dots,0\}$. It follows from Theorem \ref{thm:char_pol} that
  $A$ achieves consensus.
\end{proof}

In view of the existence Theorem \eqref{thm:existence}, consensus over
finite fields is possible on a fairly broad class of interaction
graphs. A complete characterization of all the interaction topologies
yielding consensus over finite fields is beyond the scope of this
work, and it is left as the subject of future research. We conclude
this section with a remark.

\begin{remark}{\bf \emph{(Network design for fully connected
      graphs)}}\label{remark:fully_connected}
  Let the agents interaction graph $\mc G = (\mc V , \mc E)$ be fully
  connected, that is $(i,j) \in \mc E$ for all $i,j \in
  \until{n}$. Let $v \in \subscr{\mathbb{F}}{p}^{1 \times n}$ be any
  vector satisfying $v \vectorones[] = 1$. Then the network matrix $A
  = [v^\transpose \, \cdots \, v^\transpose]^\transpose$ achieves
  consensus over $\subscr{\mathbb{F}}{p}$; see for instance the matrix
  $A_1$ in Example \ref{example:consensus}. To see this, let
  $\subscr{\vectorones[]}{orth} \in \subscr{\mathbb{F}}{p}^{n \times
    (n-1)}$ be any full column rank matrix satisfying $v
  \subscr{\vectorones[]}{orth} = 0$. Since $A \vectorones[] =
  \vectorones[]$, and $A \subscr{\vectorones[]}{orth} = 0$, we
  conclude that $\subscr{\sigma}{p} (A) = \{1,0,\dots,0\}$, with
  $|\subscr{\sigma}{p} (A)| = n$, and the claimed statement follows
  from Theorem \ref{thm:char_pol}. {\hfill\QED}
\end{remark}

\subsection{Network design via network composition}
In this section we use the concept of graph products to generate
finite-field consensus networks from smaller consensus components. For
two matrices $A \in \subscr{\mathbb{F}}{p}^{n \times n}$ and $B \in
\subscr{\mathbb{F}}{p}^{m \times m}$, let $A \otimes B \in
\subscr{\mathbb{F}}{p}^{nm \times nm}$ denote their \emph{Kronecker
  product}, where \cite{CDM:01}
\begin{align*}
   A \otimes B =
   \begin{bmatrix}
     a_{11} B & a_{12} B & \cdots & a_{1n} B\\
     a_{21} B & a_{22} B & \cdots & a_{2n} B\\
     \vdots & \vdots & \ddots & \vdots\\
     a_{n1} B & a_{n2} B & \cdots & a_{nn} B
   \end{bmatrix}.
\end{align*}

\begin{theorem}{\bf \emph{(Finite-field consensus via Kronecker
      product)}}\label{thm:kronecker}
  Consider the network matrices $A \in \subscr{\mathbb{F}}{p}^{n
    \times n}$ and $B \in \subscr{\mathbb{F}}{p}^{m \times m}$, and
  assume that $A$ and $B$ achieve consensus. Then the network
  matrix $A \otimes B \in \subscr{\mathbb{F}}{p}^{nm \times nm}$
  achieves consensus.
\end{theorem}

\begin{example}{\bf \emph{(Finite-field consensus network via
      Kronecker product)}}\label{example:kronecker}
  \begin{figure}
    \centering
    \includegraphics[width=.45\columnwidth]{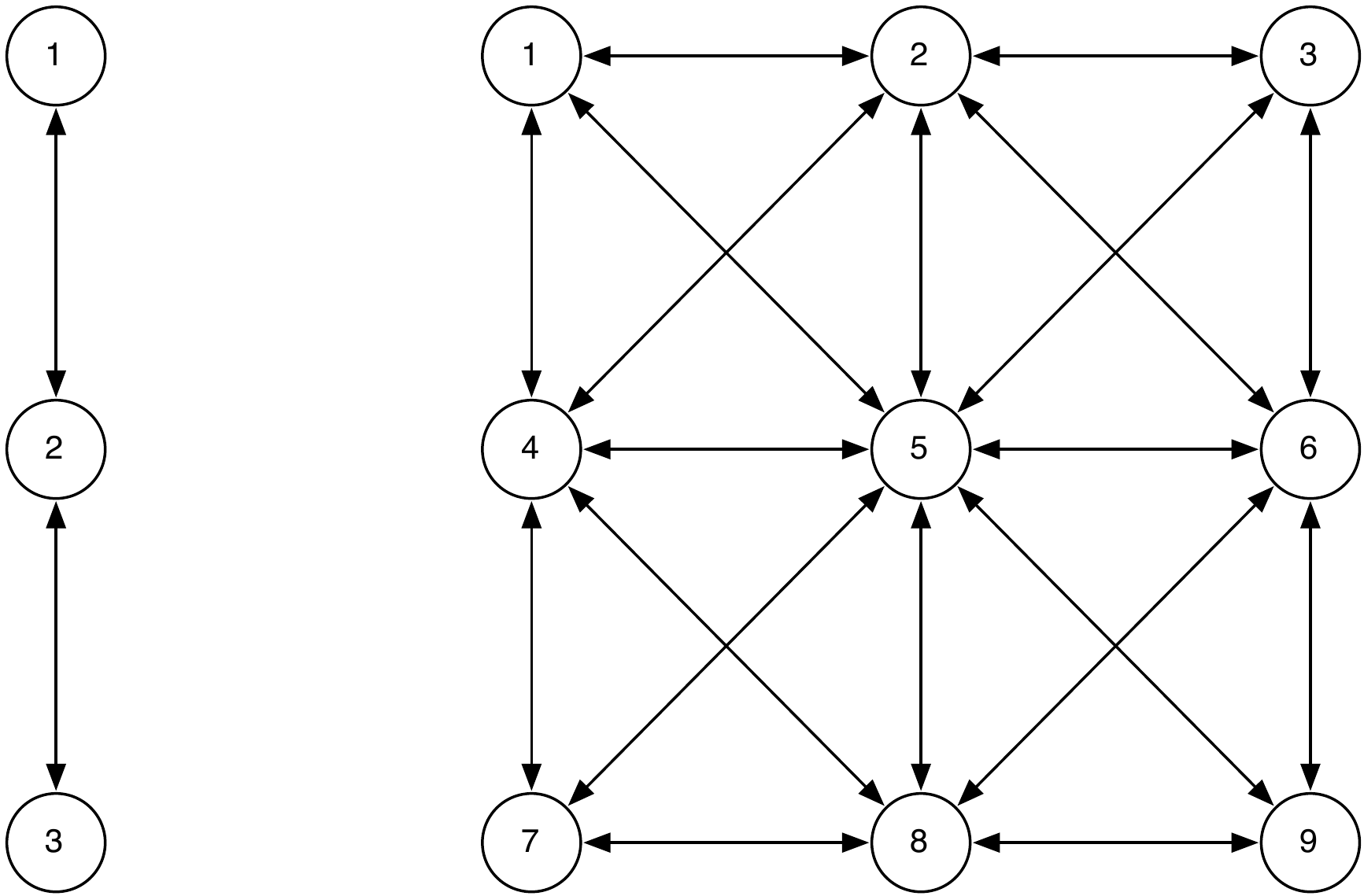}
    \caption{Agents interaction graphs for the matrices $A$ (left) and
      $\subscr{A}{k} = A \otimes A$ (right) in Example
      \ref{example:kronecker}. Self-loops have been omitted. Notice
      that the interaction graphs of $A$ and $\subscr{A}{k}$ are
      self-similar \cite{JL-DC-LK-CF-ZG:10}.}
    \label{fig:kronecker}
\end{figure}
Consider the network matrix
  \begin{align*}
    A = 
    \begin{bmatrix}
      9 & 3 & 0\\
      1 & 9 & 2\\
      0 & 7 & 5
    \end{bmatrix}
  \end{align*}
  over the field $\subscr{\mathbb{F}}{11}$. It can be verified that
  $A$ achieves consensus. By Theorem \ref{thm:kronecker} the network
  matrix
  \begin{align*}
    \subscr{A}{k} = A \otimes A =
    \begin{bmatrix}
     4&     5&     0&     5&     9&     0&     0&     0&     0\\
     9&     4&     7&     3&     5&     6&     0&     0&     0\\
     0&     8&     1&     0&    10&     4&     0&     0&     0\\
     9&     3&     0&     4&     5&     0&     7&     6&     0\\
     1&     9&     2&     9&     4&     7&     2&     7&     4\\
     0&     7&     5&     0&     8&     1&     0&     3&    10\\
     0&     0&     0&     8&    10&     0&     1&     4&     0\\
     0&     0&     0&     7&     8&     3&     5&     1&    10\\
     0&     0&     0&     0&     5&     2&     0&     2&     3
    \end{bmatrix}
  \end{align*}
  achieves consensus over $\subscr{\mathbb{F}}{11}$. In fact, it can
  be verified that $P_{\subscr{A}{k}} = s^{8} (s-1)$ and
  $\subscr{A}{k} \vectorones[] = \vectorones[]$. Moreover, the network
  matrices $A$ and $\subscr{A}{k}$ have the same convergence
  speed. The interaction graph of $A$ and $\subscr{A}{k}$ are reported
  in Fig. \ref{fig:kronecker}.  {\hfill\QED}
\end{example}

\begin{pfof}{Theorem \ref{thm:kronecker}}
  We first show that $(A \otimes B)$ is row-stochastic. Let
  $\vectorones[n]$ be the vector of all ones of dimension $n$. Since
  $A$ and $B$ are row-stochastic over $\subscr{\mathbb{F}}{p}$, we
  have
  \begin{align*}
    &(A \otimes B) \vectorones[nm] =
    \begin{bmatrix}
      a_{11} B & a_{12} B & \cdots & a_{1n} B\\
     a_{21} B & a_{22} B & \cdots & a_{2n} B\\
     \vdots & \vdots & \ddots & \vdots\\
     a_{n1} B & a_{n2} B & \cdots & a_{nn} B
    \end{bmatrix}
    \vectorones[nm]
    =
    \begin{bmatrix}
      \sum_{j = 1}^n a_{1,j} B \vectorones[m]\\
      \sum_{j = 1}^n a_{2,j} B \vectorones[m]\\
      \vdots\\
      \sum_{j = 1}^n a_{n,j} B \vectorones[m]
    \end{bmatrix}
    =
    \begin{bmatrix}
      \sum_{j = 1}^n a_{1,j}  \vectorones[m]\\
      \sum_{j = 1}^n a_{2,j}  \vectorones[m]\\
      \vdots\\
      \sum_{j = 1}^n a_{n,j} \vectorones[m]
    \end{bmatrix}
    =
    \vectorones[nm].
  \end{align*}

  Following the same argument as in \cite{CDM:01}, let $J_A = P^{-1} A
  P$ and $J_B = Q^{-1} B Q$ be the canonical Jordan form of $A$ and
  $B$, respectively \cite[Theorem 3.5]{PS:10}. Then $J_A \otimes J_B
  = (P^{-1} A P) \otimes (Q^{-1} B Q) = (P^{-1} \otimes Q^{-1}) (A
  \otimes B) (P \otimes Q) = (P \otimes Q)^{-1} (A\otimes B) (P
  \otimes Q)$, so that $A\otimes B$ and $J_A \otimes J_B$ are
  similar. Thus the eigenvalues of $A \otimes B$ are the same as those
  of $J_A \otimes J_B$, and, because $J_A$ and $J_B$ are upper
  triangular with eigenvalues $\lambda_i$ and $\mu_i$ on the diagonal,
  we conclude that $J_A \otimes J_B$ is also upper triangular with
  diagonal entries, and eigenvalues, $\lambda_i \mu_j$. Finally, since
  $\subscr{\sigma}{p} (A) = \{1,0,\dots,0\}$ and $\subscr{\sigma}{p}
  (B) = \{1,0,\dots,0\}$, the statement follows from Theorem
  \ref{thm:char_pol}.
\end{pfof}

Following Theorem \ref{thm:kronecker}, finite-field consensus networks
can be constructed by composing smaller components. We refer the
interested reader to \cite{PMW:62,JL-DC-LK-CF-ZG:10} for a
comprehensive discussion of graphs generated via Kronecker product of
adjacency matrices. Regarding the convergence speed of finite-field
consensus networks generated via Kronecker products, the following
holds. Let $\subscr{A}{k} = A_1 \otimes \dots \otimes A_m$, and let
$s_i$ be the number of iterations needed for convergence of the
network $A_i$, $i \in \until{m}$. Due to \cite[Theorem
4.3.17]{RAH-CRJ:94}, the consensus network $\subscr{A}{k}$ converges
exactly in $\max \{s_1,\dots,s_m \}$ iterations. In other words, the
consensus network $\subscr{A}{k}$ is as fast as the slowest of its
components.

\section{Application to Average Consensus and Distributed Pose
  Estimation}\label{sec:estimation}
This section contains two application scenarios for finite-field
consensus networks. In Section \ref{sec:average} we develop a
finite-time averaging algorithm based on finite-field consensus
networks. Instead, in Section \ref{sec:estimation} we use finite-field
networks to distributively estimate the orientation of a sensor
network given relative measurements.

\subsection{Finite-time average consensus}\label{sec:average}
Given a sensor network, let $x_0 \in \subscr{\mathbb{F}}{p}^n$ be the
vector containing the agents initial states. Let $x_\mathbb{R} =
\vectorones[]^\transpose x_0 / n \in \mathbb{R}$ be the average of the
agents initial states over the field of real numbers. The average of
the agents initial states over the field $\subscr{\mathbb{F}}{p}$ can
be computed by means of Fermat's little theorem \cite{SML:40} as
$x_\mathbb{F} = n^{p-2} \vectorones[]^\transpose x_0 \in \subscr{F}{p}$,
where we assume $n \neq k p$, with $k \in \mathbb{N}$, for the inverse
of $n$ over $\subscr{\mathbb{F}}{p}$ to exist. In what follows, first
we show how to compute the average $x_\mathbb{F}$ by means of
finite-field consensus networks. Then we describe conditions that
allow to recover the average $x_\mathbb{R}$ from the knowledge of
$x_\mathbb{F}$ and the total number of agents. We remark that
distributed algorithms for average consensus are useful in several
applications, including distributed parameter estimation
\cite{LX-SB-SL:05}, distributed hypothesis testing
\cite{BSYR-HFDH:93}, robotic coordination \cite{FB-JC-SM:09}, and for
the understanding of opinion dynamics \cite{MHDG:74}.

We say that the iteration \eqref{eq:autonomous} over the field
$\subscr{\mathbb{F}}{p}$ achieves average consensus if it achieves
consensus, and the consensus value is $n^{p-2}
\vectorones[]^\transpose x_0$ for every initial state $x_0$. A condition
for finite-field average consensus is given in the next theorem.

\begin{theorem}{\bf \emph{(Finite-field average
      consensus)}}\label{thm:average}
  Consider the iteration \eqref{eq:autonomous} over the field
  $\subscr{\mathbb{F}}{p}$ with row-stochastic matrix $A$. Assume that
  the field characteristic satisfies $n \neq k p$ for all $k \in
  \mathbb{N}$. The following statements are equivalent:
  \begin{enumerate}
  \item the iteration \eqref{eq:autonomous} achieves average
    consensus, and
  \item $P_A (s) = s^{n - 1}(s - 1)$, and $\vectorones[]^\transpose A
    = \vectorones[]^\transpose$.
  \end{enumerate}
\end{theorem}

\begin{example}{\bf \emph{(An example of finite-field average
      consensus)}}\label{example:average}
  \begin{figure}
    \centering
    \includegraphics[width=.9\columnwidth]{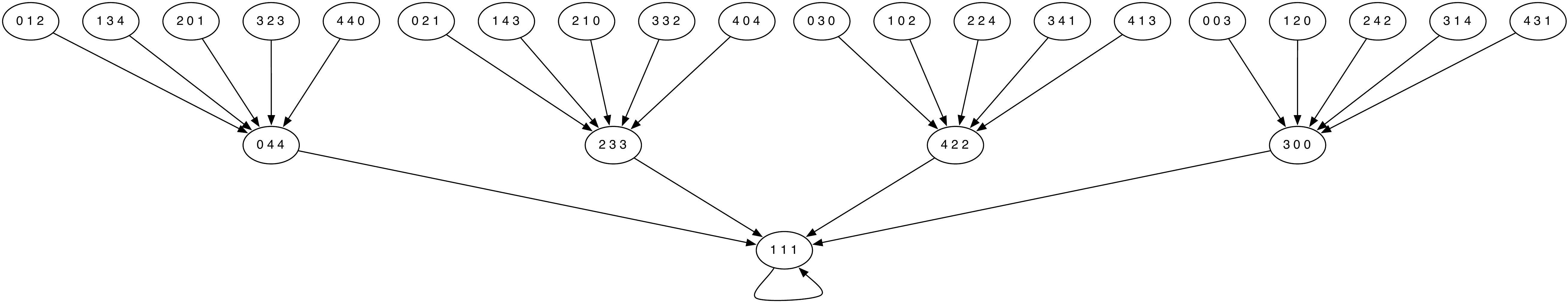}
    \caption{A subgraph of the transition graph associated with the
      network matrix $A$ in Example \ref{example:average}. Notice that
      the sum of the initial states is maintained, and thus average
      consensus is achieved.}
    \label{fig:average}
  \end{figure}
  Consider the network matrix
  \begin{align*}
    A = 
    \begin{bmatrix}
      2 & 3 & 1\\
      2 & 4 & 0\\
      2 & 4 & 0
    \end{bmatrix}
  \end{align*}
  over the field $\subscr{\mathbb{F}}{5}$. It can be verified that $A
  \vectorones[] = \vectorones[]$, $\vectorones[]^\transpose A =
  \vectorones[]^\transpose$, and $P_A = s^2(s-1)$. By Theorem
  \ref{thm:average} the network matrix $A$ achieves average consensus
  over $\subscr{\mathbb{F}}{5}$. In Fig. \ref{fig:average} we show a
  subgraph of the transition graph associated with $A$.  {\hfill\QED}
\end{example}

\begin{pfof}{Theorem \ref{thm:average}}

  \noindent \textit{(i) $\implies$ (ii)} Since the iteration achieves
  consensus, it follows from Theorem \ref{thm:convergence_value} that
  $A^n = \vectorones[] \pi$, where $\pi$ satisfies $\pi A =
  \pi$. Because $A$ achieves average consensus, it needs to be
  $\vectorones[] \pi = n^{p-2} \vectorones[]
  \vectorones[]^\transpose$. Then $\pi = n^{p-2}
  \vectorones[]^\transpose$, and $\vectorones[]^\transpose A =
  \vectorones[]^\transpose$.

  \noindent \textit{(ii) $\implies$ (i)} Since $P_A (s) = s^{n - 1}(s
  - 1)$ and $A \vectorones[] = \vectorones[]$, it follows from Theorem
  \ref{thm:char_pol} that the network achieves consensus. Notice that
  $\vectorones[]^\transpose A = \vectorones[]^\transpose$ implies that
  $\vectorones[]^\transpose x(t) = \vectorones[]^\transpose x(0)$ at
  all times time $t$. Let $\alpha$ be the consensus value, and notice
  that $n \alpha = \vectorones[]^\transpose x(0)$. To conclude the
  proof, $\alpha = n^{p-2} \vectorones[]^\transpose x(0)$, and the
  network achieves average consensus.
\end{pfof}

Theorem \ref{thm:average} provides a necessary and sufficient
condition for a network with $n \neq kp$ agents to achieve average
consensus over $\subscr{\mathbb{F}}{p}$. The condition $n \neq kp$ is
actually necessary for average consensus. In other words, if $n = kp$
for some $k \in \mathbb{N}$, then there exists no network matrix
satisfying all conditions in Theorem \ref{thm:average} and, therefore,
average consensus cannot be achieved. To see this, let $x(0)$ be the
network initial state, with $\vectorones[]^\transpose x(0) \neq 0$,
and assume by contradiction that $\alpha$ is the corresponding
consensus value. Since $\vectorones[]^\transpose x(t) =
\vectorones[]^\transpose x(0)$ at all times $t$, it needs to be $n
\alpha = \vectorones[]^\transpose x(0)$. Then $0 = n \alpha =
\vectorones[]^\transpose x(0) \neq 0$, since $kp$ and $0$ are in fact
the same element in $\subscr{\mathbb{F}}{p}$.

Suppose now that the average $x_\mathbb{F}$ has been computed, and
that each agent knows the total number of agents, the field
characteristic, and its own initial state. With these assumptions, it
is generally not possible to recover the average $x_\mathbb{R}$. To
see this, consider the case $n = 3$, $p = 5$, and the initial
conditions $x_1 = [2 \; 2 \; 2]^\transpose$ and $x_2 = [0 \; 0 \;
1]^\transpose$. Over the field of real numbers we have
$x_{1,\mathbb{R}} = \vectorones[]^\transpose x_1 / n = 2$ and
$x_{2,\mathbb{R}} = \vectorones[]^\transpose x_2 / n = 1/3$. Over the
field $\subscr{\mathbb{F}}{p}$, instead, $x_{1,\mathbb{F}} = n^{p-2}
\vectorones[]^\transpose x_1 = 2$ and $x_{2,\mathbb{F}} = n^{p-2}
\vectorones[]^\transpose x_2 = 2$. Since $x_{1,\mathbb{F}} =
x_{2,\mathbb{F}}$ and $x_{1,\mathbb{R}} \neq x_{2,\mathbb{R}}$, it is
not possible to recover the average value over the field of real
numbers from the average over a finite field and knowledge of network
cardinality and parameters. In the next theorem we present a
sufficient condition to recover the desired real-valued average.

\begin{theorem}{\bf \emph{(Average
      computation)}}\label{thm:average_computation}
  Let $x_0 \in \subscr{\mathbb{F} }{p}^n$, let $x_\mathbb{R} =
  \vectorones[]^\transpose x_0 / n \in \mathbb{R}$, and let $x_\mathbb{F}
  = n^{p-2} \vectorones[]^\transpose x_0$. If the field characteristic
  satisfies $n \| x_0 \|_\infty \le p$, then $x_\mathbb{R} =
  \text{mod} (n \, x_\mathbb{F}, p) / n$.
\end{theorem}
\begin{proof}
  The statement follows from the relation
  \begin{align*}
    \text{mod} (n \, x_\mathbb{F}, p) = \text{mod} ( n^{p-1}
    \vectorones[]^\transpose x_0 , p) = \vectorones[]^\transpose x_0,
  \end{align*}
  where the last equality holds because $\text{mod} (n^{p-1} , p) =
  1$, and $n \| x_0 \|_\infty \le p$.
\end{proof}

Our finite-time average consensus algorithm follows from Theorems
\ref{thm:average} and \ref{thm:average_computation}, and it is
reported in Algorithm \ref{algo:average}.

\begin{algorithm}[ht]
\caption{\textit{Distributed average computation (agent $i$)}}
\label{algo:average}
\begin{algorithmic}
  
  {\footnotesize
    
    \Statex \textbf{Input:} Field characteristic $p$, Initial state
    $x_i(0) \in \subscr{\mathbb{F}}{p}$, Neighbors set $\mc
    N_i^\text{in}$ and $\mc N_i^\text{out}$, Weights $a_{ij} \in
    \subscr{\mathbb{F}}{p}$ for all $j \in \mathcal{N}_i^\text{in}$, Number of
    iterations $T$ (set $T = n$ otherwise), Number of agents $n$;

    \Statex \textbf{Require:} $p$ is a prime number, $A = [a_{ij}]$
    achieves average consensus over $\subscr{\mathbb{F}}{p}$, $n \max
    \{x_1 (0), \dots, x_n(0) \} \le p$;

    \Statex \textbf{Output:} Average of initial states
    $\subscr{x}{ave} = \sum_{j = 1}^n x_j (0) /n \in \real$;

    \medskip

    \For{$t = 0,\dots,T$}

    \State{Transmit $x_i (t)$ to $\mc N_i^\text{out}$;}

    \State{Receive $x_j (t)$ from $\mc N_i^\text{in}$;}

    \State{Set $x_i (t+1) = \sum_{j \in \mathcal N^\text{in}_i} a_{ij}
      x_j(t)$;}

    \EndFor 

    \Statex {\textbf{return} $\subscr{x}{ave} = \text{mod} (n \, x_i(T) , p)/n$;}
 }
 \end{algorithmic}
\end{algorithm}

We conclude this part by noticing that the condition $n \| x_0
\|_\infty \le p$ in Theorem \ref{thm:average_computation} is not
restrictive. In fact, the field characteristic $p$ is a design
parameter, and, in general, it can be chosen to satisfy the above
condition as long as the network cardinality $n$ is known and a bound
on the agents initial state is known.

\subsection{Pose estimation from relative
  measurements}\label{sec:pose_estimation}
In this section we use our previous analysis to calibrate the
orientation of a network of cameras. This problem has been previously
considered in
\cite{GP-IS-BF-FB-BDOA:09r,WJR-DJK-JPH:11,DB-EL-RC-FF-SZ:12} as a
distributed estimation problem over $SO(2)$. With respect to the
existing literature, we let the measurements and the orientations take
value in a pre-specified finite field, and we develop an estimation
algorithm with performance guarantees based on modular arithmetic.

A camera network is modeled by an undirected graph $\mc G = (\mc V,\mc
E)$, where each vertex is associated with a camera. Let $n = |\mc V|$
and $m = |\mc E|$.
Let $\map{\theta_i}{\mathbb{N}_{\ge 0}}{ \mc O_p }$ be the orientation
of the $i$-th camera as a function of time, where, for some prime
number $p$,
\begin{align}\label{eq:angle_set}
  \mc O_p := \left\{k\frac{2\pi}{p}: k \in \fromto{0}{p-1}\right\}.
\end{align}
We refer to $p$ as to \emph{discretization accuracy}. For notational
convenience, we define the directed graph $\subscr{\mc G}{d} =
(\subscr{\mc V}{d}, \subscr{\mc E}{d})$ associated with the camera
network $\mc G$, where $\subscr{\mc V}{d} = \mc V$, and $(i,j) \in
\subscr{\mc E}{d}$ if and only if $(i,j) \in \mc E$ and $i < j$. For
each $(i,j) \in \subscr{\mc E}{d}$, let $\eta_{ij} \in \mc O_p$ be the
\emph{relative measurement} between camera $i$ and camera $j$, that is
$\eta_{ij} = \theta_i - \theta_j$. Let $\theta$ be the vector of the
cameras orientations, and let $\eta$ be the vector of relative
measurements. Assign an arbitrary ordering to the edges $\subscr{\mc
  E}{d}$, and define the incidence matrix $B \in
\subscr{\mathbb{F}}{p}^{m \times n}$ of $\subscr{\mc G}{d}$ by
specifying the $k$-row of $B$ corresponding to the edge $(i,j)$ as
\begin{align}\label{eq:incidence}
  b_{k \ell} = 
  \begin{cases}
    1, & \text{if $\ell = i$,}\\
    -1, \qquad & \text{if $\ell = j$,}\\
    0, & \text{otherwise}.
  \end{cases}
\end{align}
Observe that $B \theta = \eta$. We consider the following problem.

\begin{problem}{\bf \emph{(Pose estimation from relative
      measurements)}}\label{problem:estimation}
  Let $\mc G$ be a camera network, and let $\subscr{\mc G}{d}$ be its
  associated directed graph. Let $p$ be the discretization accuracy,
  let $B$ be the incidence matrix of $\subscr{\mc G}{d}$, and let
  $\eta \in \mc O_p^m$ be the vector of relative
  measurements. Determine a set of cameras orientations $\theta \in
  \mc O_p^n$ satisfying $B \theta = \eta$ over
  $\subscr{\mathbb{F}}{p}$.
\end{problem}

Notice that, if the incidence matrix $B$ and the relative measurements
$\eta$ are available to some camera or central unit, then Problem
\ref{problem:estimation} requires the solution of system of linear
equations. Instead, we propose an algorithm that requires each camera
to have access to local relative measurements and to communicate with
its immediate neighbors. Our distributed pose estimation algorithm is
in Algorithm \ref{algo:estimation}.

\begin{algorithm}[ht]
\caption{\textit{Distributed pose estimation (camera $i$)}}
\label{algo:estimation}
\begin{algorithmic}
  
  {\footnotesize
    
    \Statex \textbf{Input:} Discretization accuracy $p$, Initial pose
    $\theta_i(0) \in \mc O_p$, Neighbors set $\mc N_i^\text{in}$ and
    $\mc N_i^\text{out}$, Weights $a_{ij} \in
    \subscr{\mathbb{F}}{p}^{n \times n}$ for all $j \in \mc N_i^\text{in}$,
    Number of iterations $T$ (set $T = n$ otherwise), Relative
    measurements $\eta_{ij} \in \mc O_p$ for all $j \in \mc
    N_i^\text{in}$;

    \Statex \textbf{Require:} $p$ is a prime number, $A = [a_{ij}]$
    achieves average consensus over $\subscr{\mathbb{F}}{p}$;

    \Statex \textbf{Output:} Orientation $\theta_i$ compatible with
    measurements $\eta = [\eta_{ij}]$;

    \medskip

    \For{$t = 0,\dots,T$}

    \State{Transmit $x_i (t) = \frac{p \theta_i (t)}{2 \pi}$ to $\mc
      N_i^\text{out}$;}

    \State{Receive $x_j (t) = \frac{p \theta_j (t)}{2 \pi}$ from $\mc
      N_i^\text{in}$;}

    \State{Update orientation $\theta_i (t)$ as:}

    \begin{subequations}
      \begin{align}
        x_i (t+1) &= \sum_{j \in \mathcal N^\text{in}_i} a_{ij}
        \left( x_j (t)
          \label{eq:algorithm}
          + \frac{p \,\eta_{ij}}{2\pi} \right),\\
        \theta_i (t+1) &= x_i (t+1) \frac{2\pi}{p} .\;
      \end{align}
    \end{subequations}

    \EndFor 

    \Statex{\textbf{return} Orientation $\theta_i$;}

 }
 \end{algorithmic}
\end{algorithm}

Because cameras transmit and operate only on values in the finite
field $\subscr{\mathbb{F}}{p}$, we argue that Algorithm
\ref{algo:estimation} is suitable for agents with limited
capabilities, and it is robust to transmission noise. In the next
theorem we analyze the convergence of Algorithm \ref{algo:estimation}.

\begin{theorem}{\bf \emph{(Convergence of Algorithm
      \ref{algo:estimation} with perfect measurements)}}\label{thm:algorithm}
  Let $\mc G = (\mc V, \mc E)$ be a camera network, and let
  $\subscr{\mc G}{d}$ be its associated directed graph. Let $\eta \in
  \mc O_p^m$ be the vector of relative measurements, and let $B$ be
  the incidence matrix of $\subscr{\mc G}{d}$. If $\eta \in \Image(B)$
  and $A$ achieves average consensus over $\subscr{\mathbb{F}}{p}$,
  then
  \begin{enumerate}
  \item Algorithm \ref{algo:estimation} converges in finite time, that
    is, $\tilde \theta := \theta (T) = \theta (T + \tau)$ for some
    $\tilde \theta \in \mc O_p^n$, $T < n$, and for all $\tau \in
    \mathbb{N}$, and
  \item the final network orientation is compatible with the relative
    measurements, that is, $B \tilde \theta = \eta$.
  \end{enumerate}
\end{theorem}

\begin{proof} 
  Consider the update law \eqref{eq:algorithm}, and notice that it can
  be written as $x(t+1) = A x(t) + L B v$, where $\frac{p
    \,\eta}{2\pi} = Bv$ for some vector $v \in
  \subscr{\mathbb{F}}{p}^n$ ($y \in \Image(B)$ by assumption), $L \in
  \subscr{\mathbb{F}}{p}^{n \times m}$, and the $k$-th column of $L$
  corresponding to the edge $(i,j) \in \subscr{\mc E}{d}$ is specified
  as
  \begin{align}\label{eq:incidence}
    l_{\ell k} = 
    \begin{cases}
      a_{ij}, & \text{if $\ell = i$,}\\
      -a_{ij}, \qquad & \text{if $\ell = j$,}\\
      0, & \text{otherwise}.
    \end{cases}
  \end{align}
  Observe that
  \begin{align*}
    (L B)_{ij} =
    \begin{cases}
      \sum_{k \in \mathcal{N}_i^\text{in}  }  a_{i k}, \qquad & \text{if $i = j$,}\\
      - a_{ij}, & \text{if $j \in \mathcal{N}_i^\text{in} $,}\\
      0, & \text{otherwise.}
    \end{cases}
  \end{align*}
  so that $LB \vectorones[] = 0$ (asymmetric Laplacian matrix of $\mc G$
  \cite{FB-JC-SM:09}). Notice that $x(t) = A^t x(0) + \sum_{\tau =
    0}^{t-1} A^\tau L Bv$. Since $A$ achieves average consensus, we
  have $A^T = n^{p-2} \vectorones[] \vectorones[]^\transpose$ for
  some $T < n$. Thus, $A^{t} L B = n^{p-2} \vectorones[]
  \vectorones[]^\transpose L B$ for all $t \ge T$. We now show that
  $\vectorones[]^\transpose L B = 0$, from which statement (i)
  follows. Since $A$ achieves average consensus, it follows from
  Theorem \ref{thm:average} that $A \vectorones[] = \vectorones[]$ and
  $\vectorones[]^\transpose A = \vectorones[]^\transpose$. Hence, for each
  node $k \in \mc V$, $\sum_{j = 1}^n a_{k j} = \sum_{j = 1}^n a_{j
    k}$, and, consequently, $\vectorones[]^\transpose L B = 0$.

  Let $\tilde x$ be a fixed point of \eqref{eq:algorithm}, that is,
  $(I -A) \tilde x = L B v$. Since $A \vectorones[] = \vectorones[]$,
  it follows that $I = \text{diag}( \sum_{j} a_{1,j},\dots, \sum_{j}
  a_{n,j})$, and $I - A = L B$. Then $L B (\tilde x - v) = 0$ for
  every fixed point $\tilde x$ of \eqref{eq:algorithm}. Because $A$ is
  a consensus matrix, $1$ is a simple eigenvalue of $A$, and $A
  \vectorones[] = \vectorones[]$. Then $\Ker (L B) = \Ker(I - A) =
  \Image (\vectorones[])$. Notice that, by construction, $\eta =
  \frac{2 \pi}{p}B v$, $\tilde \theta = \frac{2 \pi}{p} \tilde x$, and
  $B \vectorones[] = 0$. Consequently, $(\tilde \theta - \eta) \in
  \Image (\vectorones[])$, and statement (ii) follows.
\end{proof}

In Theorem \ref{thm:algorithm} we assume that the measurements satisfy
$\eta \in \Image (B)$, or, equivalently, that the measurements are not
affected by noise. While this assumption is justified by the fact that
we only consider discretized measurements, in what follows we study
the evolution of Algorithm \ref{algo:estimation} when the measurements
are affected by noise. Let $e (t) = \eta - B \theta(t)$.

\begin{theorem}{\bf \emph{(Convergence of Algorithm
      \ref{algo:estimation} with noisy measurements)}}\label{thm:algorithm_noise}
  Let $\mc G = (\mc V, \mc E)$ be a camera network, and let
  $\subscr{\mc G}{d}$ be its associated directed graph. Let $\eta \in
  \mc O_p^m$ be the vector of relative measurements, and let $B$ be
  the incidence matrix of $\subscr{\mc G}{d}$. If $A$ achieves average
  consensus over $\subscr{\mathbb{F}}{p}$, then
  \begin{enumerate}
  \item there exists a finite time $T < n$ such that the estimation
    error is constant, that is, $e(t) = e (t + 1)$ for all $t \ge T$,
    and
  \item for all $t \ge T$, the estimation error satisfies
    \begin{align*}
      e (t) = \left( I - B \sum_{\tau =0}^{T - 1} A^{\tau} L \right) \subscr{\eta}{orth},
    \end{align*}
    where $\subscr{\eta}{orth}$ is the orthogonal projection of $\eta$
    onto $\Image(B)^\perp$, and $L$ is as in \eqref{eq:incidence}.
  \end{enumerate}
\end{theorem}
\begin{proof}
  With the same notation as in the proof of Theorem
  \ref{thm:algorithm}, let $T < n$ be the number of iterations needed
  for convergence of the network matrix $A$, that is, $A^T = n^{p-2}
  \vectorones[]\vectorones[]^\transpose$. Notice that $B A^t = 0$ for
  all $t \ge T$, so that
  \begin{align*}
    e(t) &= \eta - A^t x(0) - \sum_{\tau = 0}^{t - 1} A^{\tau} L \eta
    = \eta - \sum_{\tau = 0}^{T - 1} A^{\tau} L \eta,
  \end{align*}
  for all $t \ge T$, and statement (i) follows. To show statement
  (ii), let $\eta = \subscr{\eta}{par} + \subscr{\eta}{orth}$, where
  $\subscr{\eta}{par} \in \Image(B)$, and $\subscr{\eta}{orth} \in
  \Image(B)^\perp$. From the linearity of \eqref{eq:algorithm} and
  Theorem \ref{thm:algorithm} we have
  \begin{align*}
    \subscr{\eta}{par} = B \sum_{\tau = 0}^{T - 1} A^{\tau} L
    \subscr{\eta}{par},
  \end{align*}
  which concludes the proof.
\end{proof}

Theorem \ref{thm:algorithm_noise} characterize the performance of
Algorithm \ref{algo:estimation} when the measurements are affected by
noise. Notice that the estimation error can be minimized by properly
choosing the network matrix $A$. We now conclude with a numerical
example.


\begin{example}{\bf \emph{(Average computation and pose estimation via
      finite-field consensus)}}\label{example_numerical}
  Consider a camera network with $4$ cameras configured in a circle
  topology
  and network matrix
\begin{align*}
  A = 
  \begin{bmatrix}
    0  &   4  &   2 &    0\\
    1  &   1  &   0 &    4\\
    0  &   0  &   2 &    4\\
    0  &   1  &   2 &    3
  \end{bmatrix}
  \in \subscr{\mathbb{F}}{5}^{4 \times 4} .
\end{align*}
It can be verified that $A$ achieves average consensus over
$\subscr{\mathbb{F}}{5}$ in at most $3$ iterations. In
Fig. \ref{fig:convergence_average} we show that the network matrix $A$
allows for the computation of the real-valued average of the agents
initial states in $3$ iterations (see Algorithm \ref{algo:average} and
Section \ref{sec:analysis}).

Let $\subscr{A}{k} \in \subscr{\mathbb{F}}{5}^{1024 \times 1024}$ be
the network matrix generated from $A$ as $\subscr{A}{k} = A \otimes A
\otimes A \otimes A \otimes A$. The sparsity pattern of
$\subscr{A}{k}$ is reported in Fig. \ref{fig:sparsity}. In
Fig. \ref{fig:pose_estimation} we validate our distributed pose
estimation algorithm (see Section \ref{sec:pose_estimation}).
{\hfill\QED}
\end{example}

\begin{figure}[tb]
  \centering \subfigure[]{
    \includegraphics[width=.45\columnwidth]{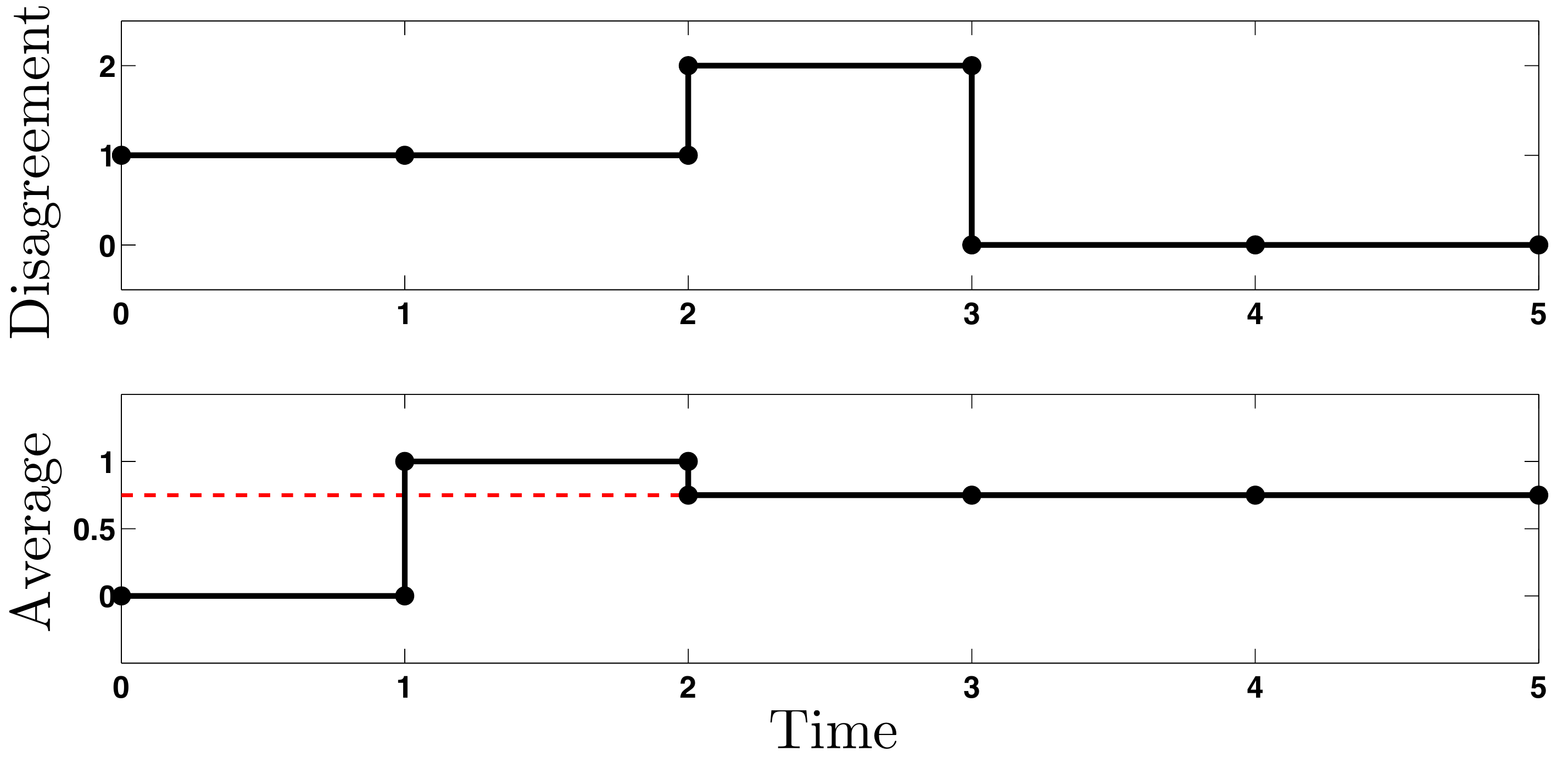}
    \label{fig:convergence_average}
  } \subfigure[]{
    \includegraphics[width=.45\columnwidth]{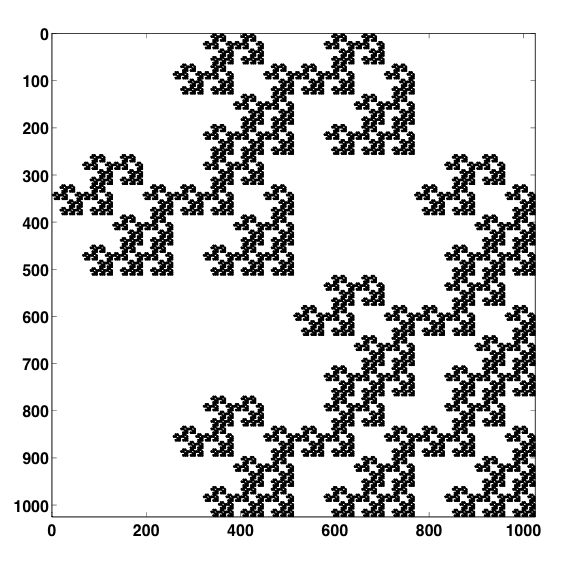}
    \label{fig:sparsity}
  }
  \caption[Optional caption for list of figures]{For the network
    described by the matrix $A$ in Example \ref{example_numerical},
    let $0,1,1,1$ be the agents initial states, respectively. Agents
    implement Algorithm \ref{algo:average}. In Fig.
    \ref{fig:convergence_average} (above), we report the network
    disagreement as a function of time, where the disagreement at a
    given time equals the largest agents state minus the smallest
    agents state. Notice that the network achieves consensus in $3$
    iterations (see Section \ref{sec:analysis}). In Fig.
    \ref{fig:convergence_average} (below), we report the average
    computed by the first agent (solid black) as a function of time
    (see Section \ref{sec:average}). Notice that the first agent, and
    hence every agent in the network, computes the average of the
    initial states (dashed red) at the third iteration. %
    Fig. \ref{fig:sparsity} shows the sparsity pattern of the matrix
    $\subscr{A}{k} \in
      \subscr{\mathbb{F}}{5}^{1024 \times 1024}$ in Example
      \ref{example_numerical}. The network defined by $\subscr{A}{k}$
      has $1024$ nodes and $10^5$ edges.}
  \label{fig:sub1}
\end{figure}



\begin{figure}
    \centering
    \includegraphics[width=.6\columnwidth]{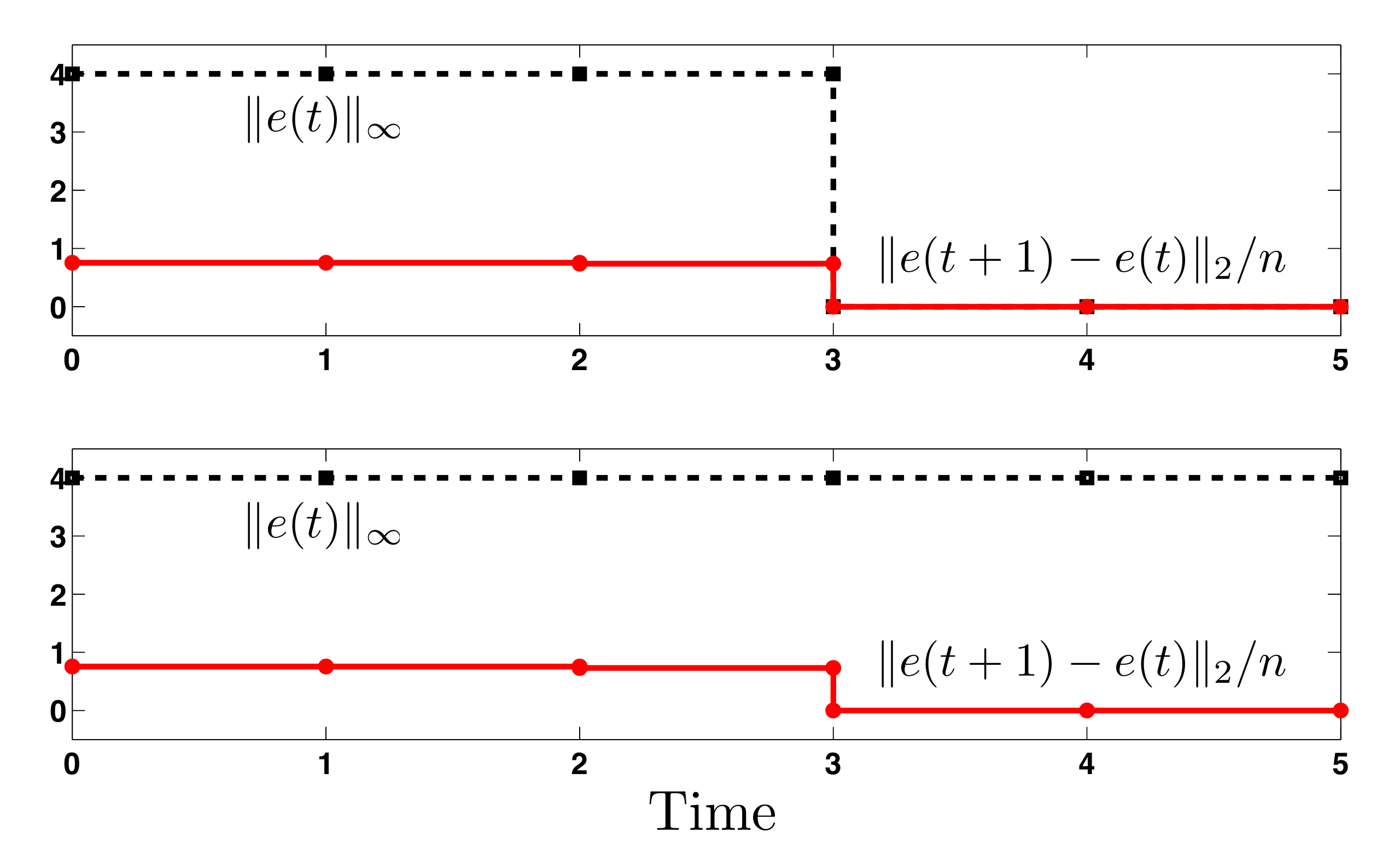}
    \caption{For the camera network with $1024$ cameras described by
      the matrix $\subscr{A}{k}$ in Example \ref{example_numerical},
      this figure show the effectiveness of the our pose estimation
      algorithm in Algorithm \ref{algo:estimation}. The figure above
      considers the noiseless case (Theorem \ref{thm:algorithm}). The
      dashed black line denotes the infinity norm of the estimation
      error, while solid red line corresponds to the (normalized) two
      norm of the difference between two consecutive estimation
      errors. We conclude that, in the absence of measurement noise,
      Algorithm \ref{algo:estimation} converges in $3$ iterations to a
      configuration compatible with the relative measurements. The
      figure below considers the case of noisy measurements (Theorem
      \ref{thm:algorithm_noise}). Notice that, since the measurements
      are affected by noise, the estimation error remains nonzero
      (dashed black). However, the algorithm converges to a
      configuration with constant estimation error. Indeed, the
      difference between two consecutive estimation errors is zero
      after $3$ iterations (solid red).}
    \label{fig:pose_estimation}
\end{figure}

\section{Conclusion and Future Work}\label{sec:conclusion}
In this paper we propose a distributed consensus algorithm for agents
with limited memory, computation, and communication capabilities. Our
approach is based on finite-fields, where agents states lie in a
finite set, and operations are performed according to modular
arithmetic. For our algorithm we identify necessary and sufficient
convergence conditions, and we characterize the convergence
time. Additionally, we discuss several network design methods, and we
propose some application scenarios. Our work proposes a novel class of
consensus dynamics, which are advantageous in several applications,
and it complements the existing literature on real-valued consensus.

Through our analysis we show that finite-field consensus networks
outperforms their real-valued counterpart in many aspects, including
the convergence speed, robustness to communication errors, and agents
requirements. These advantages come at the expenses of a more
convoluted network design, which we identify as an interesting
research direction. In particular, distributed network design
algorithms, as well as gossip and asynchronous protocols should be
investigated to broaden the applicability of finite-field consensus
networks. Other theoretical research questions include characterizing
the existence of finite-field consensus weights for a given
interconnection structure and field characteristic, as well as the
design of fastest finite-field consensus networks with fixed
agents interconnection graph.

\renewcommand{\theequation}{A-\arabic{equation}}
\setcounter{equation}{0}  
\section*{APPENDIX}  

Before proving Theorem \ref{thm:char_pol}, we recall the following
fundamental results and facts in linear algebra.

\begin{theorem}{\bf \emph{(Primary decomposition
      theorem \cite{GES:77})}}\label{thm:primary}
  Let $\map{A}{\mc V}{\mc V}$ be a linear operator on some vector
  space $\mc V$ over some field $\mathbb{F}$, and let $p(s) =
  \prod_{i=1}^r p_i (s)$ be an annihilating polynomial for $A$ with
  degree greater than $1$, for some relatively prime polynomials
  $p_1,\dots,p_r$. Then
  \begin{enumerate}
  \item $\mc W_i = \Ker (p_i (A))$ is a $A$-invariant subspace for all $i
    \in \until{r}$,
  \item $\mc V = \mc W_1 \oplus \mc W_2 \oplus \dots \oplus \mc W_r$,
    where $\oplus$ denote the direct sum operator, and
  \item if $\prod_{i=1}^r p_i (s)= P_A (s)$ and $A_i$ is the restriction of
    $A$ to $\mc W_i$, then $p_i$ is the characteristic polynomial of
    $A_i$.
  \end{enumerate}
\end{theorem}

Recall that the order of a polynomial $g \in \mathbb{F}[s]$, denoted
by $\text{ord}(g)$, is the smallest positive integer $r$ such that
$g(s)$ divides $s^r - 1$ over $\mathbb{F}$, that is, the smallest
positive integer $r$ such that there exists $q\in\mathbb{F}[s]$
satisfying $s^r - 1 = g(s) q(s)$.

\begin{theorem}\emph{\textbf{(Order of a polynomial over finite field
      \cite{EB:59})}}\label{lemma: order_poly}  Let $g\in
  \mathbb{F} [s]$ be an irreducible polynomial satisfying
  $g(0)\neq 0$ and $\text{ord}(g) = e$. Consider $f = g^s$, and $t$ is
  the smallest integer such that $p^t \geq s$, then $\text{ord}(f) = e p^t$.
\end{theorem}

\smallskip
We are now ready to prove Theorem \ref{thm:char_pol}.

\begin{pfof}{Theorem \ref{thm:char_pol}}
  Let $P_A \in \subscr{\mathbb{F}}{p}[s]$ be the characteristic
  polynomial of $A$, and notice that $P_A$ can be written as
  \begin{align}\label{eq: decomp}
      P_A(s) = \text{det}(sI - A) = s^h \bar P(s),
  \end{align}
  for some $h \in \mathbb{N}_{\ge 0}$, and $\bar
  P(s)\in\subscr{\mathbb{F}}{p}[s]$ with $\bar P(0) \neq 0$.

  \noindent
  \textit{(i) $\implies$ (ii)} 
  Due to Lemma \ref{lemma: stochastic}, we have $1 \in
  \subscr{\sigma}{p} (A)$. Thus, we factorize $P_A$ in irreducible
  polynomials as
  \begin{align*}
    P_A (s) = (s - 1)^k \prod_{j = 1}^r Q_j(s)^{m_j},
  \end{align*}
  where $k,m_j\in\mathbb{N}$ are given by the algebraic multiplicity
  of the corresponding eigenvalue.

  We start by showing that $k = 1$. Assume by contradiction that $k >
  1$.  Let $\mc W_2 = \Ker ((I - A)^k)$, and let $A_2$ be the
  restriction of $A$ to $\mc W_2$. Recall from \cite{RAHT:05} that the
  cycle structure of the transition graph $\mc G_2$ of $A_2$ is
  \begin{align}\label{eq:cycles}
    \text{Cycles} (\mc G_2) = C_1 + \sum_{i = 1}^k (p^i - p^{i - 1}) C_{1},
  \end{align}
  where the sums of cycles is just the corresponding union graph, and
  $C_1$ denotes a unitary cycle, that is, a fixed point for $A$. From
  \eqref{eq:cycles} it follows that, if $k > 1$, then the number of
  cycles in $\mc G_2$ is strictly greater than $p$. By Theorem
  \ref{thm:transition} we conclude that $k = 1$.

  We now show that $r = 0$. Assume by contradiction that $r > 0$. Let
  $\mc W_3 = \Ker (Q_j(A)^{m_j})$, and let $A_3$ be the restriction of
  $A$ to $\mc W_3$. Then the cycle structure of the transition graph
  $\mc G_3$ of $A_3$ is
  \begin{align*}
    \text{Cycles} (\mc G_3) = C_1 + \sum_{i = 1}^{m_j}
    \frac{p^{\textup{deg} (Q_j) i} - p^{\textup{deg} (Q_j) (i -
        1)}}{\ell_i} C_{\ell_i},
  \end{align*}
  where $\ell_i = \textup{ord}(Q_j^{n_j}) \geq \textup{deg} (Q_j) \geq
  1$ from Theorem \ref{lemma: order_poly}, and $\textup{deg}(\cdot)$
  denotes the degree of a polynomial. Since the graph structure of $A$
  is given by the product of the graphs associated with the
  irreducible factors of its characteristic polynomial \cite{RAHT:05},
  the number of cycles is greater than $p$ whenever either $k > 1$ or
  $r > 0$ (see Example \ref{example:char_pol} and
  Fig. \ref{fig:transition3}).

  \noindent
  \textit{(ii) $\implies$ (i)} Let $\mc W_1 = \Ker (A - I) =
  \Image(\vectorones[])$ and recall from Theorem \ref{thm:primary}
  that $\mc W_1$ is $A$-invariant. Let $V = [V_1 \, \vectorones[]]$ be
  an invertible matrix, where the columns of $V_1$ are a basis for
  $\mc W_1^\perp$. Then we have
  \begin{align*}
    \tilde A = V^{-1} A V =
    \begin{bmatrix}
      A_{11} & 0\\
      A_{21} & 1
    \end{bmatrix}.
  \end{align*}
  Since the eigenvalues of a matrix are not affected by similarity
  transformations, the characteristic polynomial of the matrix
  $A_{11}$ is $s^{n-1}$, so that $A_{11}$ is nilpotent. It follows
  that every vector in $\mc W_1^\perp$ converges to the origin in at
  most $n-1$ iterations, while vectors in $\mc W_1$ (consensus
  vectors) are fixed points for the matrix $A$. This concludes the
  proof.
\end{pfof}

\bibliographystyle{unsrt}
\bibliography{alias,Main,FB}

\begin{thebibliography}{10}

\bibitem{NAL:97}
N.~A. Lynch.
\newblock {\em Distributed Algorithms}.
\newblock Morgan Kaufmann, 1997.

\bibitem{FB-JC-SM:09}
F.~Bullo, J.~Cort{\'e}s, and S.~Mart{\'\i}nez.
\newblock {\em Distributed Control of Robotic Networks}.
\newblock Applied Mathematics Series. Princeton University Press, 2009.

\bibitem{MM-ME:10}
M.~Mesbahi and M.~Egerstedt, editors.
\newblock {\em Graph Theoretic Methods in Multiagent Networks}.
\newblock Princeton University Press, 2010.

\bibitem{FG-LS:10}
F.~Garin and L.~Schenato.
\newblock A survey on distributed estimation and control applications using
  linear consensus algorithms.
\newblock In A.~Bemporad, M.~Heemels, and M.~Johansson, editors, {\em Networked
  Control Systems}, LNCIS, pages 75--107. Springer, 2010.

\bibitem{WR-RWB-EMA:07}
W.~Ren, R.~W. Beard, and E.~M. Atkins.
\newblock Information consensus in multivehicle cooperative control:
  {C}ollective group behavior through local interaction.
\newblock {\em {IEEE} Control Systems Magazine}, 27(2):71--82, 2007.

\bibitem{LX-SB-SL:05}
L.~Xiao, S.~Boyd, and S.~Lall.
\newblock A scheme for robust distributed sensor fusion based on average
  consensus.
\newblock In {\em Symposium on Information Processing of Sensor Networks},
  pages 63--70, Los Angeles, CA, USA, April 2005.

\bibitem{DPB-JNT:97}
D.~P. Bertsekas and J.~N. Tsitsiklis.
\newblock {\em Parallel and Distributed Computation: Numerical Methods}.
\newblock Athena Scientific, 1997.

\bibitem{RL-HN:96}
R.~Lidl and H.~Niederreiter.
\newblock {\em Finite Fields}.
\newblock Cambridge University Press, 1996.

\bibitem{YGS-LW-GX:08}
Y.~G. Sun, L.~Wang, and G.~Xie.
\newblock Average consensus in networks of dynamic agents with switching
  topologies and multiple time-varying delays.
\newblock {\em Systems \& Control Letters}, 57(2):175--183, 2008.

\bibitem{TCA-MEY-ADS-AS:09}
T.~C. Aysal, M.~E. Yildiz, A.~D. Sarwate, and A.~Scaglione.
\newblock Broadcast gossip algorithms for consensus.
\newblock {\em IEEE Transactions on Signal Processing}, 57(7):2748--2761, 2009.

\bibitem{SK-JMF:09}
S.~Kar and J.~M.~F. Moura.
\newblock Distributed consensus algorithms in sensor networks with imperfect
  communication: Link failures and channel noise.
\newblock {\em IEEE Transactions on Signal Processing}, 57(1):355--369, 2009.

\bibitem{LM:05}
L.~Moreau.
\newblock Stability of multiagent systems with time-dependent communication
  links.
\newblock {\em IEEE Transactions on Automatic Control}, 50(2):169--182, 2005.

\bibitem{RC-FB:06j}
R.~Carli and F.~Bullo.
\newblock Quantized coordination algorithms for rendezvous and deployment.
\newblock {\em SIAM Journal on Control and Optimization}, 48(3):1251--1274,
  2009.

\bibitem{AN-AO-AO-JNT:09}
A.~Nedi{\'c}, A.~Olshevsky, A.~Ozdaglar, and J.~N. Tsitsiklis.
\newblock On distributed averaging algorithms and quantization effects.
\newblock {\em IEEE Transactions on Automatic Control}, 54(11):2506--2517,
  2009.

\bibitem{RC-FF-PF-SZ:10}
R.~Carli, F.~Fagnani, P.~Frasca, and S.~Zampieri.
\newblock Gossip consensus algorithms via quantized communication.
\newblock {\em Automatica}, 46(1):70--80, 2010.

\bibitem{TL-MF-LX-JFZ:11}
T.~Li, M.~Fu, L.~Xie, and J.~F. Zhang.
\newblock Distributed consensus with limited communication data rate.
\newblock {\em IEEE Transactions on Automatic Control}, 56(2):279--292, 2011.

\bibitem{JL-RMM:12}
J.~Lavaei and R.~M. Murray.
\newblock Quantized consensus by means of gossip algorithm.
\newblock {\em IEEE Transactions on Automatic Control}, 57(1):19--32, 2012.

\bibitem{AF-EMV-AB:08}
A.~Fagiolini, E.~M. Visibelli, and A.~Bicchi.
\newblock Logical consensus for distributed network agreement.
\newblock In {\em {IEEE} Conf.\ on Decision and Control}, pages 5250--5255,
  Canc\'un, M\'exico, December 2008.

\bibitem{AF-ND-AB:11}
A.~Fagiolini, N.~Dubbini, and A.~Bicchi.
\newblock Distributed consensus on set-valued information, 2011.
\newblock Available at \texttt{http://arxiv.org/abs/1101.2275}.

\bibitem{AK-TB-RS:07}
A.~Kashyap, T.~Ba{\c s}ar, and R.~Srikant.
\newblock Quantized consensus.
\newblock {\em Automatica}, 43(7):1192--1203, 2007.

\bibitem{AO:12}
A.~Olshevsky.
\newblock Consensus with ternary messages, 2012.
\newblock Available at \texttt{http://arxiv.org/abs/1212.5768}.

\bibitem{SS-CH:12}
S.~Sundaram and C.~Hadjicostis.
\newblock Structural controllability and observability of linear systems over
  finite fields with applications to multi-agent systems.
\newblock {\em IEEE Transactions on Automatic Control}, 58(1):60--73, 2013.

\bibitem{RK-MM:03}
R.~Koetter and M.~M{\'e}dard.
\newblock An algebraic approach to network coding.
\newblock {\em IEEE/ACM Transactions on Networking}, 11(5):782--795, 2003.

\bibitem{EB:59}
B.~Elspas.
\newblock The theory of autonomous linear sequential networks.
\newblock {\em IRE Transactions on Circuit Theory}, 6(1):45--60, 1959.

\bibitem{GES:77}
G.~Shilov.
\newblock {\em Linear Algebra}.
\newblock {N}ew {Y}ork: {D}over {P}ublications, 1977.

\bibitem{CDG-GFR:01}
C.~D. Godsil and G.~F. Royle.
\newblock {\em Algebraic Graph Theory}, volume 207 of {\em Graduate Texts in
  Mathematics}.
\newblock Springer, 2001.

\bibitem{RAHT:05}
R.~A.~H. Toledo.
\newblock Linear finite dynamical systems.
\newblock {\em Communications in Algebra}, 33(9):2977--2989, 2005.

\bibitem{WR-RWB:05}
W.~Ren and R.~W. Beard.
\newblock Consensus seeking in multi-agent systems under dynamically changing
  interaction topologies.
\newblock {\em IEEE Transactions on Automatic Control}, 50(5):655--661, 2005.

\bibitem{JCD-RC-SZ:07}
J.~C. Delvenne, R.~Carli, and S.~Zampieri.
\newblock Optimal strategies in the average consensus problem.
\newblock In {\em {IEEE} Conf.\ on Decision and Control}, New Orleans, USA,
  December 2007.

\bibitem{CDM:01}
C.~D. Meyer.
\newblock {\em Matrix Analysis and Applied Linear Algebra}.
\newblock SIAM, 2001.

\bibitem{PS:10}
P.~Singla.
\newblock On representations of general linear groups over principal ideal
  local rings of length two.
\newblock {\em Journal of Algebra}, 324(9):2543--2563, 2010.

\bibitem{ROS-RMM:03c}
R.~Olfati-Saber and R.~M. Murray.
\newblock Consensus problems in networks of agents with switching topology and
  time-delays.
\newblock {\em IEEE Transactions on Automatic Control}, 49(9):1520--1533, 2004.

\bibitem{ASF-YY:79}
A.~S. Fraenkel and Y.~Yesha.
\newblock Complexity of problems in games, graphs and algebraic equations.
\newblock {\em Discrete Applied Mathematics}, 1(1-2):15--30, 1979.

\bibitem{MRG-DSJ:79}
M.~R. Garey and D.~S. Johnson.
\newblock {\em Computers and Intractability}.
\newblock Springer, 1979.

\bibitem{JD-JEG-DS:06}
J.~Ding, J.~E. Gower, and D.~Schmidt.
\newblock {\em Multivariate Public Key Cryptosystems}.
\newblock Springer, 2006.

\bibitem{LB-JCF-LP:09}
L.~Bettale, J.~C. Faug{\`e}re, and L.~Perret.
\newblock Hybrid approach for solving multivariate systems over finite fields.
\newblock {\em Journal of Mathematical Cryptology}, 3(3):177--197, 2009.

\bibitem{NC-AK-JP-AS:00}
N.~Courtois, A.~Klimov, J.~Patarin, and A.~Shamir.
\newblock Efficient algorithms for solving overdefined systems of multivariate
  polynomial equations.
\newblock In {\em Advances in Cryptology -- EUROCRYPT 2000}, volume 1807 of
  {\em Lecture Notes in Computer Science}, pages 392--407. Springer, 2000.

\bibitem{JL-DC-LK-CF-ZG:10}
J.~Leskovec, D.~Chakrabarti, J.~Kleinberg, C.~Faloutsos, and Z.~Ghahramani.
\newblock Kronecker graphs: {A}n approach to modeling networks.
\newblock {\em The Journal of Machine Learning Research}, 11:985--1042, 2010.

\bibitem{PMW:62}
P.~M. Weichsel.
\newblock The {K}ronecker product of graphs.
\newblock {\em Proceedings of the American Mathematical Society}, 13(1):47--52,
  1962.

\bibitem{RAH-CRJ:94}
R.~A. Horn and C.~R. Johnson.
\newblock {\em Topics in Matrix Analysis}.
\newblock Cambridge University Press, 1994.

\bibitem{SML:40}
S.~Mac Lane.
\newblock Modular fields.
\newblock {\em The American Mathematical Monthly}, 47(5):259--274, 1940.

\bibitem{BSYR-HFDH:93}
B.~S.~Y. Rao and H.~F. Durrant-Whyte.
\newblock A decentralized {B}ayesian algorithm for identification of tracked
  targets.
\newblock {\em IEEE Transactions on Systems, Man \& Cybernetics},
  23(6):1683--1698, 1993.

\bibitem{MHDG:74}
M.~H. DeGroot.
\newblock Reaching a consensus.
\newblock {\em Journal of the American Statistical Association},
  69(345):118--121, 1974.

\bibitem{GP-IS-BF-FB-BDOA:09r}
G.~Piovan, I.~Shames, B.~Fidan, F.~Bullo, and B.~D.~O. Anderson.
\newblock On frame and orientation localization for relative sensing networks.
\newblock {\em Automatica}, 49(1):206--213, 2013.

\bibitem{WJR-DJK-JPH:11}
W.~J. Russell, D.~J. Klein, and J.~P. Hespanha.
\newblock Optimal estimation on the graph cycle space.
\newblock {\em IEEE Transactions on Signal Processing}, 59(6):2834 --2846,
  2011.

\bibitem{DB-EL-RC-FF-SZ:12}
D.~Borra, E.~Lovisari, R.~Carli, F.~Fagnani, and S.~Zampieri.
\newblock Autonomous calibration algorithms for networks of cameras.
\newblock In {\em {A}merican {C}ontrol {C}onference}, pages 5126--5131,
  Montr\'eal, Canada, June 2012.

\end{thebibliography}

\end{document}